\documentclass[12pt,letterpaper]{article}

\usepackage{amsmath}
\usepackage{amsfonts}
\usepackage{amssymb}
\usepackage{graphicx}
\usepackage{color}
\usepackage{graphicx}
\usepackage{epstopdf}
\usepackage{amsmath}
\usepackage{amsfonts}
\usepackage{amssymb}
\usepackage{color}
\usepackage{mathrsfs}
\usepackage{color}
\usepackage{enumerate}

\usepackage[utf8]{inputenc}

\usepackage{cancel,dsfont} 
\pdfoutput=1
\usepackage{bbm}

\usepackage{mathrsfs}
\usepackage{amsfonts, amssymb}
\usepackage{graphicx}
\usepackage{psfrag}
\usepackage[usenames,dvipsnames,svgnames,table]{xcolor}
\usepackage{enumerate}

\usepackage{xcolor,url}
\usepackage{cancel,dsfont}

\usepackage[colorlinks,bookmarks]{hyperref}
\definecolor{linkblue}{rgb}{0,0,0.8}
\definecolor{linkgreen}{rgb}{0,0.5,0}

\hypersetup{pdfpagemode=UseNone, pdfstartview=FitH, linkcolor=linkblue,
            citecolor=linkgreen, urlcolor=linkblue}

\usepackage{amsthm}

\newcommand\eea{\end{eqnarray}}
\newcommand\bea{\begin{eqnarray}}

\def\be{\begin{equation}}
\def\ee{\end{equation}}

\newtheorem{theorem}{Theorem}

\newtheorem{lemma}{Lemma}
\newtheorem{claim}{Claim}

\theoremstyle{definition}
\newtheorem{definition}{Definition}
\newtheorem{assumption}{Assumption}
\newtheorem{example}{Example}
\newtheorem{remark}{Remark}

\newcommand{\eps}{\varepsilon}
\begin{document}

\begin{center}

{\Large  \bf

A {de Sitter} no-hair theorem for 3+1d Cosmologies \\[0.1cm] {with isometry group forming 2-dimensional orbits}  
\\[0.7cm]}
{\large  Paolo Creminelli${}^{1}$, Or Hershkovits${}^{2}$, 
Leonardo Senatore${}^{3}$, Andr\'as Vasy${}^{2}$ }
\\[0.7cm]

{\normalsize { \sl $^{1}$Abdus Salam International Centre for Theoretical Physics\\ Strada Costiera 11, 34151, Trieste, Italy\\[0.1cm]
IFPU - Institute for Fundamental Physics of the Universe,\\
Via Beirut 2, 34014, Trieste, Italy}}\\
\vspace{.4cm}

{\normalsize { \sl $^{2}$ Department of Mathematics,\\ Stanford University, Stanford, CA 94306}}\\
\vspace{.4cm}

{\normalsize { \sl $^{3}$ Stanford Institute for Theoretical Physics,\\Department of Physics, Stanford University, Stanford, CA 94306}}\\
\vspace{.2cm}

{\normalsize { \sl 
Kavli Institute for Particle Astrophysics and Cosmology, \\
Department of Physics and SLAC, Stanford University, Menlo Park, CA 94025}}\\
\vspace{.3cm}

\vspace{.3cm}

\end{center}

\vspace{.8cm}

\hrule \vspace{0.3cm}
{\small  \noindent \textbf{Abstract}} \\[0.3cm]
\noindent 
We study, using Mean Curvature Flow methods, 3+1 dimensional cosmologies with a positive cosmological constant, matter satisfying the dominant and the strong energy conditions, and with spatial slices that can be foliated by {2-dimensional} surfaces that are the closed orbits of a symmetry group. If these surfaces have non-positive Euler characteristic (or in the case of 2-spheres, if the initial 2-spheres are large enough) and also if the initial spatial slice is expanding everywhere, then we prove that asymptotically the spacetime becomes  physically indistinguishable from de Sitter space on arbitrarily large regions of spacetime. This holds true notwithstanding the presence of initial arbitrarily-large density fluctuations. 
 
 \vspace{0.3cm}
\hrule


 \vspace{0.3cm}
\newpage
\tableofcontents

\section{Introduction}

Inflation is widely believed to be a cosmological epoch that occurred before the epoch  of radiation dominance (the hot big bang). Typically, it is driven by a scalar field that runs down its flat potential, homogeneously and slowly, and leads to an exponential expansion of the universe. Inflation seems to be required to produce an approximately flat homogeneous and isotropic universe endowed with small perturbations that, in the theory of inflation, are due to the quantum fluctuations of  the scalar field while it rolls down. Inflation has been extraordinarily successful when compared with observational data from the Cosmic Microwave Background (see for example~\cite{Aghanim:2018eyx,Akrami:2018odb,Akrami:2019izv}) or from the Large-Scale Structure of the universe (see for example~\cite{DAmico:2019fhj,Ivanov:2019pdj,Colas:2019ret,Abbott:2017wau,Abbott:2018wzc}). Despite all these observational successes, the onset of inflation has been a source of heated debate
for a long time. If a region of space somewhat larger than the Hubble length during inflation is homogeneously filled with the inflationary scalar field at the top of its potential, then inflation starts, but the debate is about how likely it is for the universe to have such a homogenous initial condition. This is the so-called `initial patch problem' (see for example~\cite{Ijjas:2015hcc}). 

Solid progress on this matter was hard to achieve because the presence of large inhomogeneities and the formation of singularities made it hard to attack the problem both numerically and analytically, at least without imposing symmetries. Recently, however, there  has been significant progress on both fronts. Initially, on the numerical side, the codes that can handle singularities and that are normally used in the prediction of the templates of gravitational waves from black-hole mergers~\cite{Pretorius:2005gq} have been applied to simulate the early universe. Ref.~\cite{East:2015ggf}, and subsequently~\cite{Clough:2016ymm,Clough:2017efm}, have found {numerical} evidence that, on an extremely large set of inhomogenous initial conditions, inflation always starts. On the analytical side, a combination of Mean Curvature Flow techniques (see for example~\cite{gerhardtbook}) and the now-proven Thurston Geometrization Classification (see \cite{besse1987einstein} Theorem 4.35 and \cite{thurston1997three,10.2307/2152760}) allowed to prove some partial results in the general case, without imposing extra symmetries. In particular, approximating the inflationary potential as a positive cosmological constant, and assuming that matter satisfies the weak energy condition and that all singularities are of the  so-called  crushing kind, Ref.~\cite{Kleban:2016sqm} has shown that, for almost all topologies of the spatial slices of a cosmological spacetime, the volume of these slices (assumed to be, initially, expanding everywhere) will grow with time (see also~\cite{barrow1985closed}); moreover, there is always an open neighborhood that expands at least as fast as the flat of de Sitter space. This suggests, though does not prove, that the volume will go to  infinity, matter will dilute away, and the universe will resemble de Sitter space  in arbitrarily large regions of spacetime. This statement was recently proven  in 2+1 dimensions (with the additional assumption that matter satisfies the strong and the dominant energy condition~\cite{Creminelli:2019pdh}; see~\cite{Barrow:2006cw} for proofs with stronger assumptions on the matter content and on the initial conditions). Historically, it has been conjectured for many years and with different level of refinment (see for instance \cite{Gibbons:1977mu,Hawking:1981fz,Wald:1983ky,Kleban:2016sqm}) that in the presence of a positive cosmological constant, cosmologies that are initially ``sufficiently expanding" should asymptote to de Sitter space. This is usually dubbed the de Sitter no-hair conjecture.

In this paper we focus on 3+1 dimensions, and we assume that the spatial slices can be foliated by {2-dimensional} surfaces that are the closed orbits of a symmetry group (in addition to the assumptions just discussed for the theorem in 2+1 dimensions). We will find that asymptotically in the future, the spacetime appears physically indistinguishable from de Sitter space, { in the following sense}. {Any future observers will have at their disposal a vanishing amount of energy and momentum to make any experiment. Furthermore,} the length of any future-directed timelike or null curve approaches the one computed with the de Sitter metric {(see Theorem \ref{main_thm_intro} for the full statement, and Section \ref{equiv_de_sit_sec} for a physical explanation of why the mathematical results imply that, asymptotically, the spacetime is physically indistinguishable from de Sitter, {in a low energy sense})}. In the context of 3+1 dimensions stronger convergence results were obtained in~\cite{TchapndaN.:2003bv,BlaiseTchapndaN.:2003fv,Tchapnda:2007uy,LeFloch:2010jg,Andreasson:2013pga}, assuming more symmetries and prescribing specific PDEs which govern the matter stress tensor (from point particles to stiff fluids). Assuming homogeneity of the \textit{entire} spatial slices (while here we assume homogeneity only on 2-dimensional slices), Wald proved pointwise convergence to de Sitter assuming the strong and the dominant energy condition for matter~\cite{Wald:1983ky} for all Bianchi universe except type IX. 

It is important to stress that a de Sitter no-hair theorem is also a statement about the asymptotic future of the present universe, assuming that the present acceleration is due to a cosmological constant. 

{Let us mention that from the geometric standpoint,  our result fits into an extensive body of literature of studying the structure of spaces satisfying some curvature conditions, using special submanifolds. Such special submanifolds could be geodesics (as in the Bonnet-Myers theorem~\cite{BonnetMyers}), minimal surfaces (as in the proof of the positive mass theorem \cite{SchoenYau}) or submanifolds produced by some curvature flows (as in the proofs of the Riemannian Penrose inequality~\cite{Hui_Ilm} and of the high co-dimensional isoperimetric inequality for surfaces \cite{Schulze}). In our setting, the curvature conditions imposed by the Einstein equation and the energy conditions are reminiscent of a lower Ricci curvature bound - a topic which has been studied in depth in the works of Cheeger, Colding, Naber and others (c.f. \cite{CheegerGromoll,CheegerColding, ColdingNaber, CheegerNaber}).}  \\      

{We have tried to write this paper in a way that would be approachable to both the cosmology and the geometric analysis communities. We have therefore decided to spell out many derivations which are ``standard'' in one discipline, for the benefit of the other community.}

\section{General assumptions and known results}\label{general_ass_sec}

We will prove a theorem that uses some properties of the topology of 3-dimensional manifolds, as well as of the mean curvature flow. It requires the following assumptions~\cite{Creminelli:2019pdh}, on top of others that we will specify next:

\begin{enumerate}[(A)]

\item There is  a ``cosmology'', which is defined as a connected $3+1$ dimensional spacetime $M^{(3+1)}$ with a compact Cauchy surface.  This implies that the spacetime is topologically  $ M^{(3)}\times \mathbb{R}$ where $M^{(3)}$ is a compact $3$-manifold, and that it can be foliated by a family of topologically identical Cauchy surfaces $M_t$ \cite{Geroch}. We fix one such foliation,  {\it {\it i.e.}}~such a time function $t$, with $t\in[t_0,+\infty)$, and with associated lapse function $N$: $N^{-2} := -\partial_\mu t \partial^\mu t$, $N >0$. We consider manifolds that are initially expanding everywhere, {\it {\it i.e.}}~there is an initial slice,  $M_{0}$, where $K>0$ everywhere, with $K$ being the mean curvature with respect to the future pointing normal to $M_0$. For example, $K>0$ holds if one has a global crushing singularity in the past.

{\item $M^{(3+1)}$ satisfies Einstein's field equation 
\be\label{eq:Einstein_cond}
R_{\mu\nu}-\frac12 g_{\mu\nu} R = 8 \pi G_N (T_{\mu\nu}-\Lambda g_{\mu\nu}) \;,
\ee
where, $R_{\mu\nu}$ is the Ricci curvature tensor, $R$ is the scalar curvature,  $\Lambda$ is the cosmological constant and $T_{\mu\nu}$ is the stress-energy tensor of all the other forms of matter.  
}
\item There is a positive cosmological constant and matter that satisfies the Dominant Energy Condition (DEC) and the Strong Energy Condition (SEC). The DEC states that $-T^\mu{}_\nu k^\nu$ is a future-directed timelike or null vector for any future-directed timelike vector $k^\mu$. The DEC implies the Weak Energy Condition (WEC),  $T_{\mu \nu}k^{\mu}k^{\nu} \geq 0$ for all time-like vectors~$k^{\mu}$. The SEC, in $3+1$ dimensions, reads: $(T_{\mu\nu}- \frac{1}{2}g_{\mu\nu} T) k^\mu k^\nu  \geq 0$ for any future-directed timelike vector $k^\mu$.

\item\label{only_crush}  We will also need a technical assumption, see Definition~\ref{def:crushing}: the only spacetime singularities are of the crushing kind~\cite{Eardley} (thus singularities that have zero spatial volume). Physically, these are the only singularities that are believed to be relevant.

\end{enumerate}

Let us comment on the physical restrictions implied by the above hypotheses. The SEC and the DEC are satisfied by non-relativistic matter, radiation and the gradient energy of a scalar field \footnote{For SEC indeed
\be
T_{\mu\nu} = \partial_\mu\phi \partial_\nu\phi - \frac12 g_{\mu\nu} (\partial\phi)^2 
\quad\Rightarrow\quad  \left(T_{\mu\nu}- \frac{g_{\mu\nu}}{2} T\right) k^\mu k^\nu = (\partial \phi \cdot k)^2 \geq 0\;.
\ee
}. The Inflationary potential violates SEC and if the potential is negative somewhere also DEC is violated. However, in our setup the Inflationary potential is represented by the positive cosmological constant, which is a good approximation in the inflationary region of the potential. 

We also comment on the definition of a crushing singularity, as we follow~\cite{Creminelli:2019pdh} in adopting a slight generalization of the Definitions 2.10 and 2.11 in~\cite{Eardley}. Our definition will agree with theirs in the case of asymptotically flat spacetimes.

\begin{definition}\label{def:crushing} Analogously to Definition 2.9 of~\cite{Eardley}, a future crushing function $\tilde t$ is a globally defined function on $M^{(3+1)}$ such that on a globally hyperbolic neighborhood ${\cal{N}}\cap\{\tilde t>c_0\}$, $\tilde t$ is a Cauchy time function with range $c_0<\tilde t<+\infty$ ($c_0 \ge 0$ is a constant), and such that the level sets~$S_c=\{\tilde t=c\}$, with $c>c_0$, have mean curvature $\tilde K<-c$.~\footnote{For example in a Schwarzschild-de Sitter spacetime in the standard coordinates, one could take $\tilde t$ to be a function of $r$ for $r$ close to 0, so the level sets $S_c$ would be $r={\rm const}$.}  We shall say that a Cosmology has potential singularities only of the crushing kind if there is an open set ${\cal{N}}$ such that, outside ${\cal{N}}$, the inverse of the lapse of the $t$ foliation, $N^{-1}$, is bounded, and such that ${\cal{N}}$ contains a Cauchy slice and admits a future crushing function $\tilde t$ and, for any given $c$, in $\{\tilde t \leq c\}$, $N^{-1}$ is bounded. 
\end{definition}

\noindent In physical terms,  this ${\cal{N}}$ corresponds to a subset of the interior of black holes, and we are requiring that any possible pathology takes place only for $\tilde t\to \infty$. \\

{Choosing any $c_1\geq c_0$ that we later specify, we define a new time function on $M^{(3+1)}$, which we call $t$ from now on, such that the lapse $N$ is set to 1 in the region where $\tilde t\leq c_1$. In this region, the new time function $t$ now satisfies $\partial_\mu t\partial^\mu t=-1$.}\\

We will use the Mean Curvature Flow (MCF) of codimension-one spacelike surfaces in Lorentzian manifolds. This is defined as the deformation of a slice as follows: $y^\mu(\cdot,\lambda):= y_\lambda$ is,  at each $\lambda$, a mapping between the initial spatial manifold ${{M}}_0$ (which is parametrized by $x$), and the global spacetime, ${{M}}_0 \times [\lambda_{\rm in},\lambda_{\rm max}) \to M^{(3+1)}$. We take $\lambda_{\rm in}=0$. The evolution under the change of $\lambda$ is given by (see for instance \cite{EH})
\be
\frac{d}{d\lambda}y^\mu(x,\lambda)=K n^\mu (y^\alpha)\ ,
\ee 
where $n^\mu$ is the future-oriented vector orthonormal to the surface of constant $\lambda$. {We denote by $\mathcal{M}_\lambda$ the geometric image of $y(\cdot,\lambda)$}.  

 \begin{figure}
\begin{center}
\includegraphics[width=11cm,draft=false]{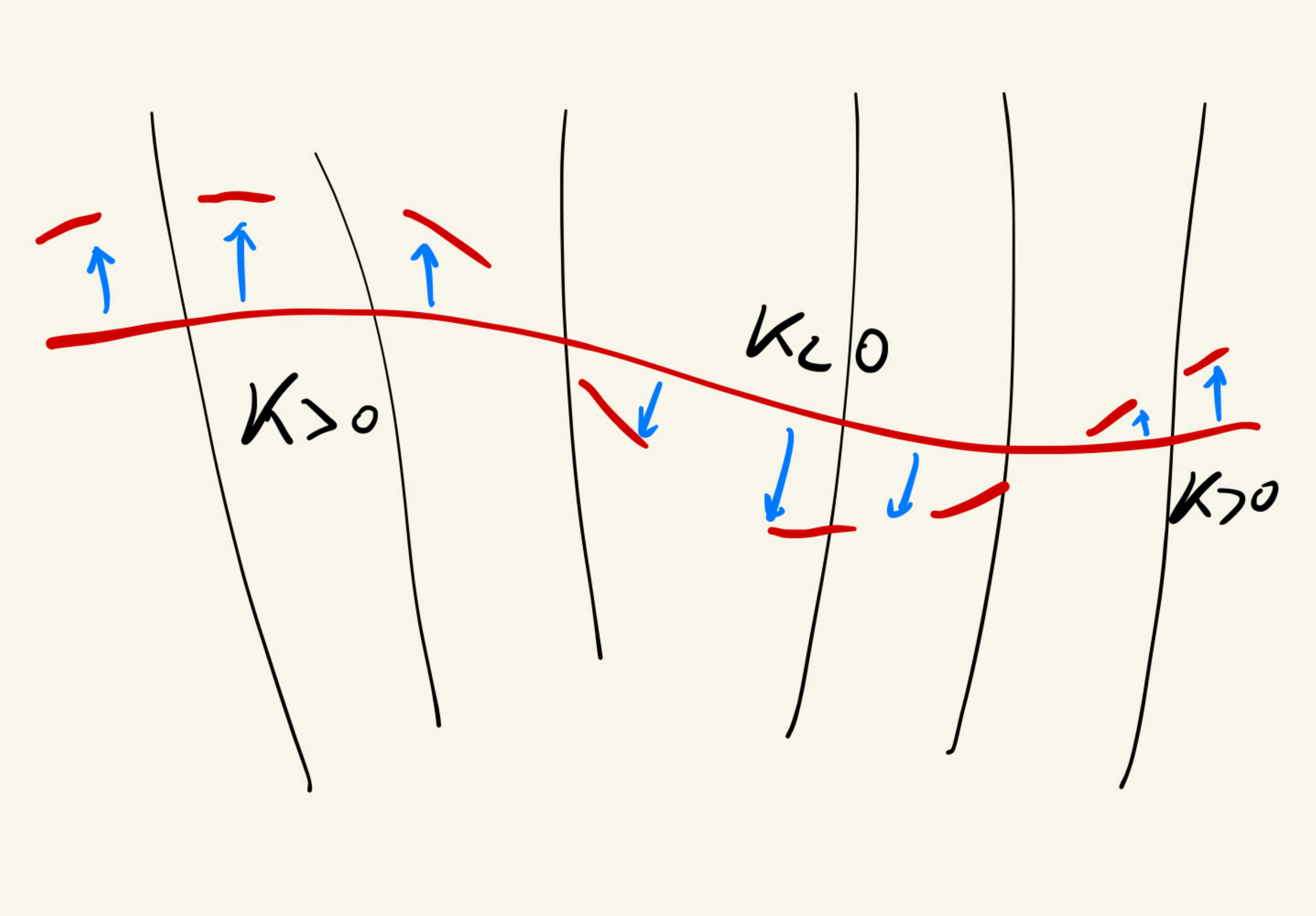}
\end{center}
\caption{\small A depiction of Mean Curvature Flow. The new surface has larger or equal volume than the previous one. \label{fig:MKF} }
\end{figure}

Using the first variation of area formula
  \be\label{normdeform}
{\cal{L}}_{n} \log\sqrt{h}=K\ ,
 \ee
one gets the variation of {the volume element} $\sqrt{h}$ under the flow: ${d \over d \lambda} \sqrt{h} = K^2 \sqrt{h}$.  Therefore  the total spatial volume 
$V{(\lambda)} := \int_{{ \mathcal{M}_\lambda}} d^4 x \sqrt{h}$ satisfies
\be \label{volgrowth}
{d V \over d \lambda} = \int_{{ \mathcal{M}_{\lambda}}} d^4 x \sqrt{h}\,K^2 \geq 0 \; .
\ee

Hence after the deformation, the new surface has either strictly larger or equal volume (see Fig.~\ref{fig:MKF}).  MCF has been very much studied in the context of Riemannian manifolds, but there is quite a large literature also for the Lorentzian (or semi-Riemannian) one, see~\cite{EH, gerhardtbook}.

We will assume that $M^{(3+1)}$ satisfies Einstein equations, and we will use MCF to probe the geometry of $M^{(3+1)}$. This is possible because the flow  is endowed by many regularity  properties as we review below. {Importantly, in the Lorentzian  cosmological context, the flow is globally graphical, which is rarely a natural assumption in the Riemannian setting}.

{
The evolution of $K$  under MCF reads
\begin{equation}\label{K_evolve}
0=\frac{dK}{d\lambda}-\Delta K+\frac{1}{3}K^3+(\sigma^2+\mathrm{Ric}(n,n))K\ ,
\end{equation}
where $\Delta$ is the Laplacian operator on the {three dimensional evolving} surface, ${\cal M}_\lambda$, where we remind that $\sigma^2$ is the norm squared of the traceless part of the second fundamental form{, and where $\mathrm{Ric}$ is the Ricci tensor} (See \cite[Proposition 3.3]{EH}). Substituting $(n,n)$ into  the Einstein equation \eqref{eq:Einstein_cond}, we get
\begin{equation}\label{play_ein1}
\mathrm{Ric}(n,n)+\frac{1}{2}R=8\pi G_{N}(T(n,n)+\Lambda)\ ,
\end{equation}   
while tracing \eqref{eq:Einstein_cond} yields
\begin{equation}\label{play_ein2}
-R=8\pi G_{N}(T-4\Lambda)\ ,
\end{equation}
where $T$ is the trace of $T_{\mu\nu}$. Combining \eqref{play_ein1} with \eqref{play_ein2} gives
\begin{equation}
\mathrm{Ric}(n,n)=-8\pi G_{N}\Lambda +8\pi  G_N\left(T(n,n)-\frac{1}{2}Tg(n,n)\right)\ ,
\end{equation}
 which, after substituting into \eqref{K_evolve} gives
\be\label{eq:MCF}
\frac{d K}{d \lambda} -\Delta K + \frac{1}{3} K\left(K^2-K_\Lambda^2\right) + \sigma^2 K + R^{(m)}_{\mu\nu} n^\mu n^\nu K =0 \;,
\ee
}
where 
\begin{equation}\label{K_lambda_def}
K_\Lambda^2 := 24\pi G_N \Lambda>0\ ,
\end{equation}  
and
\be
R^{(m)}_{\mu\nu} := 8\pi G_N \left(T_{\mu\nu}-\frac{g_{\mu\nu}}{2} T\right) \ .
\ee
The SEC gives
\be\label{eq:SEC}
R^{(m)}_{\mu\nu} n^\mu n^\nu \geq 0 \;.
\ee
Two properties of the evolution under MCF are worthwhile mentioning. First, if a surface is spacelike, it remains so: in fact the local volume form is non-decreasing under MCF, but it would vanish if the surface became null anywhere~(see for example~\cite{Kleban:2016sqm}).  Second, it also preserves the property that $K>0$ everywhere  (see \emph{e.g.} \cite{gerhardtbook}, Proposition 2.7.1).  Intuitively, this is because the flow stops in any region where $K$ approaches zero.

Our stated assumptions were used in~\cite{Creminelli:2019pdh} to prove the following useful statements about  the maximum of $K$ and the existence of the flow. We reproduce them here for convenience, referring to~\cite{Creminelli:2019pdh} for their proofs.

\begin{theorem}[Bound on the Maximum of $K$]\label{th:boundonmax}~\cite{Creminelli:2019pdh}
Let ${\cal M}_\lambda$ be smooth compact spacelike hypersurfaces satisfying the MCF equations, in an interval $[0,\lambda_1]$, inside the smooth $(3 + 1)$-dimensional Lorentzian manifold  $M^{(3+1)}$ {satisfying \eqref{eq:Einstein_cond} and SEC}. Suppose also there exists a point $(x,\lambda)$, with $0\leq\lambda\leq\lambda_1$, such that $K(x,\lambda)>K_\Lambda$, then  we have
\be\label{eq:maxbound}
K_m(\lambda_1)\leq K_\Lambda+ e^{- \frac{2}{3}K_\Lambda^2\lambda_1}(K_m(0)-K_\Lambda) \le K_\Lambda\left(1+C_1 e^{- \frac{2}{3}K_\Lambda^2\lambda_1}\right)\ .
\ee
with $C_1=\max(K_m(0)/K_\Lambda-1, 0)$. So the maximum, if larger than $K_\Lambda$, decays exponentially fast towards $K_\Lambda$ with a rate given by the cosmological constant.
\end{theorem}

Notice that if no  point $(x,\lambda)$ as in the hypotheses of the theorem exists, then the maximum $K_m(\lambda)$, with $0\leq\lambda\leq\lambda_1$, is automatically $\leq K_\Lambda$.\\

{We also have the following  long time existence theorem, which follows from~\cite{EH}, Assumption (\ref{only_crush}), and Theorem \ref{th:boundonmax}:}

\begin{theorem}[Existence of the flow]\label{theorem:existence}~\cite{Creminelli:2019pdh}
Let $M^{(3+1)}$ be a Cosmology satisfying the SEC and DEC, having potential singularities only of the crushing kind. Let $M_{0}$ be a compact smooth spacelike
hypersurface in $M^{(3+1)}$. Then there exists a unique family (${\cal M}_\lambda$) of smooth compact spacelike hypersurfaces satisfying the MCF equations with initial condition $M_{0}$, in the semi-interval $[0,+\infty)$.
\end{theorem}

\section{Symmetry assumptions}\label{symmetry_assumption_sec}

For some of the most interesting settings it is sufficient to make the following {simple symmetry assumption}.

\begin{assumption}[Simplified symmetry assumption]\label{simp_ass}
There is a Lie group $G$ which acts  on $M^{(3)}$ such that the induced action on $M^{(3+1)}$ is by isometries, and {such} that the orbits under $G$ are closed surfaces. Assume that the orbits of $G$ are two-sided ({\it i.e.}, with trivial normal bundle).  
\end{assumption}

Taking $M_0$ and considering its mean curvature flow ${\cal M}_{\lambda}$ starting from  $M_0$, we see that the isometries in $G$ preserve the level sets of $\lambda$ as well.

\begin{example}\label{torus_ex}
Consider $M^{(3)}$ being the three torus $\mathbb{T}^3=S^1\times S^1\times S^1$ such that given a point $(x,t)=(\theta_1,\theta_2,\theta_3,t)\in \mathbb{T}^3\times \mathbb{R}=M^{(3+1)}$, the metric at $(x,t)$ is independent of $\theta_1,\theta_2$. Taking $G=S^1\times S^1$, it acts on $M^{(3+1)}$ by isometries, as for every $\phi=(a,b)\in S^1\times S^1$ we can set  $\phi(x,t)=(\theta_1+a,\theta_2+b,  \theta_3,t)$.
\end{example}

\begin{example}\label{sphere_ex}
Consider $M^{(3)}$ being the product $S^2\times S^1$  and $M^{(3+1)}=S^2\times S^1\times  \mathbb{R}$. Letting $G=SO(3)$ be the group of orientation preserving orthogonal transformations of the three Euclidean space,  $G$ acts on  $M^{(3+1)}$ by $\phi(x,\theta,t)=(\phi(x),\theta,t)$, where $(x,\theta,t)\in S^2\times S^1\times \mathbb{R}$. If this action is by isometries, then assumption \ref{simp_ass} is satisfied . One can construct such a metric as follows: Let $P^+$ be the north-pole of $S^2$, and  choose a metric $h_0$ for $M^{(3+1)}$ \textit{along}  the surface $\{P^{+}\}\times S^1\times \mathbb{R}$, with signature $(+,+,+,-)$, with $\partial_t$ timelike, and such that for each $(\theta, t)$ all rotations of $S^2$ across the axis from the north to south pole are isometries of $T_{(P^+,\theta,t)}M^{(3+1)}$ (in the linear algebra sense). Then for every $x\in S^2$ choose any $\phi\in G$ s.t. $\phi(x)=P^{+}$ and let $h|_{(x,\theta,t)}=\phi^{*}h|_{(P^{+},\theta,t)}$.  
\end{example}

The drawback of working only under the simplified assumption above is that it imposes that the compact $2$-dimensional orbits of $G$ have a transitive isometry group acting on them.  {Compact surfaces of negative curvature however have only discrete isometry groups}. Thus, hyperbolic surfaces, which are ``most surfaces'' in some sense, can not arise under the above assumption. To overcome this we assume:

\begin{assumption}[General symmetry assumption]\label{full_ass}
There is a Lie group $\tilde{G}$ which acts  on some cover $\pi:\tilde{M}^{(3)}\rightarrow M^{(3)}$, such that the induced action on $\tilde{M}^{(4)}:=\tilde{M}^{(3)}\times \mathbb{R}$ is by isometries. Assume further that the orbits of $\tilde{G}$ are two dimensional complete {({\it i.e.}~with no edges)} surfaces, and that for each such orbit $\tilde{\Sigma}$, its projection $\Sigma:=\pi(\tilde{\Sigma})$ is a two-sided {surface}.
\end{assumption}

Now, letting $\tilde{{M_0}}=\pi^{-1}({M_0})$, and $\tilde{{\cal M}}_{\lambda}=\pi^{-1}({\cal M}_{\lambda})$, we see that $\tilde{{\cal{M}}}_{\lambda}$ is a MCF with bounded curvatures and height over finite intervals,  {emanating} from $\tilde{{M_0}}$. {As every $\phi\in \tilde{G}$ is an isometry of $\tilde{M}^{(4)}$, and as $\phi(\tilde{{M_0}})=\tilde{{M_0}}$,  we have that $\phi(\tilde{{\cal M}}_{\lambda})$ is also a MCF, with bounded curvatures and height, emanating from $\tilde{{M_0}}$}. Standard  uniqueness theory (see \cite{uniqueness_ref}\footnote{{The results in \cite{uniqueness_ref} are about MCF in an ambient \textit{Riemannian manifold} with bounded $\|\nabla^k\mathrm{Rm}\|$ for $k=0,1,2$, and we are unaware of a reference where such a uniqueness result is stated in the Lorentzian setting. In our setting we already have a bound on the motion by Theorem \ref{th:boundonmax}, and so everything occurs inside a covering preimage of compact set. This, combined with \cite[Theorem 4.4]{EH} (which is valid in our non-compact setting because of periodicity) implies that all the geometric quantities in the relevant analysis will be bounded. Arguing similarly to \cite{uniqueness_ref} we will get uniqueness in our setting. See also \cite{uniqueness_ref_ric}.}}) gives that $\phi(\tilde{{\cal M}}_{\lambda})=\tilde{{\cal M}}_{\lambda}$. 

In particular, if $(\tilde{x},t)\in \tilde{{\cal M}}_\lambda$ then for every $\phi\in \tilde{G}$, $(\phi(\tilde{x}),t)\in\tilde{{\cal M}}_\lambda$. Letting $\tilde{\Sigma}$ be the orbit of such $(\tilde{x},t)$, we see that along $\tilde{\Sigma}$, all intrinsic and extrinsic geometric quantities are invariant under the action of $\tilde{G}$. Thus, all intrinsic and extrinsic scalar quantities on $\Sigma=\pi(\tilde{\Sigma})$ (such as $K,{}^{(2)}\!R,{}^{(3)}\!R,H_{\mu\nu}H^{\mu\nu},|K_{ij}|^2,|\nabla K|^2, |\sigma_{ij}|^2$ in Section \ref{sec:asymptotics_geometrical}) are constant along it.  

{
\begin{example}\label{hyper_ex}
Let $\Gamma$ be a discrete co-compact subgroup of $\tilde{G}:=O(2,1)$- the group of isometries of the hyperbolic plane $\mathbb{H}^2$, and consider  $M^{(3)}=\left(\mathbb{H}^2/\Gamma\right) \times S^1$. $\tilde{G}$ acts on  $\tilde{M}^{(4)}$ by $\phi(x,\theta,t)=(\phi(x),\theta,t)$, so any metric on $M^{(3+1)}$ such that  $\tilde{G}$ acts by isometries on its pull-back to $\tilde{M}^{(4)}$ will satisfy Assumption~\ref{full_ass}. To obtain such a metric, we can use a similar construction to the one in Example~\ref{sphere_ex}.
\end{example} 
}
 In addition to Examples \ref{torus_ex}, \ref{sphere_ex} and \ref{hyper_ex}, other examples include the topologies   $\mathbb{T}^{2}/\Gamma\times S^{1}$ (with $\Gamma$ being a freely acting finite subgroup of isometries of two-torus) and  $S^{2}\times_{\mathbb{Z}_2} S^{1}$. The case of $\mathbb{H}^{3}/\Gamma$ (with $\Gamma$ being a discrete, co-compact subgroup of isometries of $\mathbb{H}^{3}$) does not fit into our setting (but it does fit into the one in \cite{Wald:1983ky}). {A pictorial representation of an example of the geometry of the spatial slices allowed by our assumptions is given in Fig.~\ref{fig:example_manifold}.}

\begin{remark}
It is interesting to compare our symmetry assumption in the  torus case of example \ref{torus_ex} with the symmetry assumptions of previous results on the no-hair conjecture for a coupled evolution rule of the metric and the stress energy tensor (such as the Einstein-Vlasov system) (c.f.~\cite{TchapndaN.:2003bv,BlaiseTchapndaN.:2003fv,Tchapnda:2007uy,LeFloch:2010jg,Andreasson:2013pga}).  Prior to \cite{Andreasson:2013pga}, all results assumed a full $3$-dimensional group of symmetries (corresponding either to completely homogeneous 3-dimensional spaces \cite{Wald:1983ky} or to homogeneous \textit{and isotropic} cross-sectional surfaces). In \cite{Andreasson:2013pga}, a so called $\mathbb{T}^3$-Gowdy symmetry was imposed, {and in fact, results indicating some asymptotic resemblance to de Sitter were obtained there  for \textit{general} matter satisfying some energy conditions.}  The group of $\mathbb{T}^3$-Gowdy symmetries imposes a few additional discrete symmetries on top of the $S^1\times S^1$ symmetries we impose. 
\end{remark}

 \begin{figure}
\begin{center}
\includegraphics[width=11cm,draft=false]{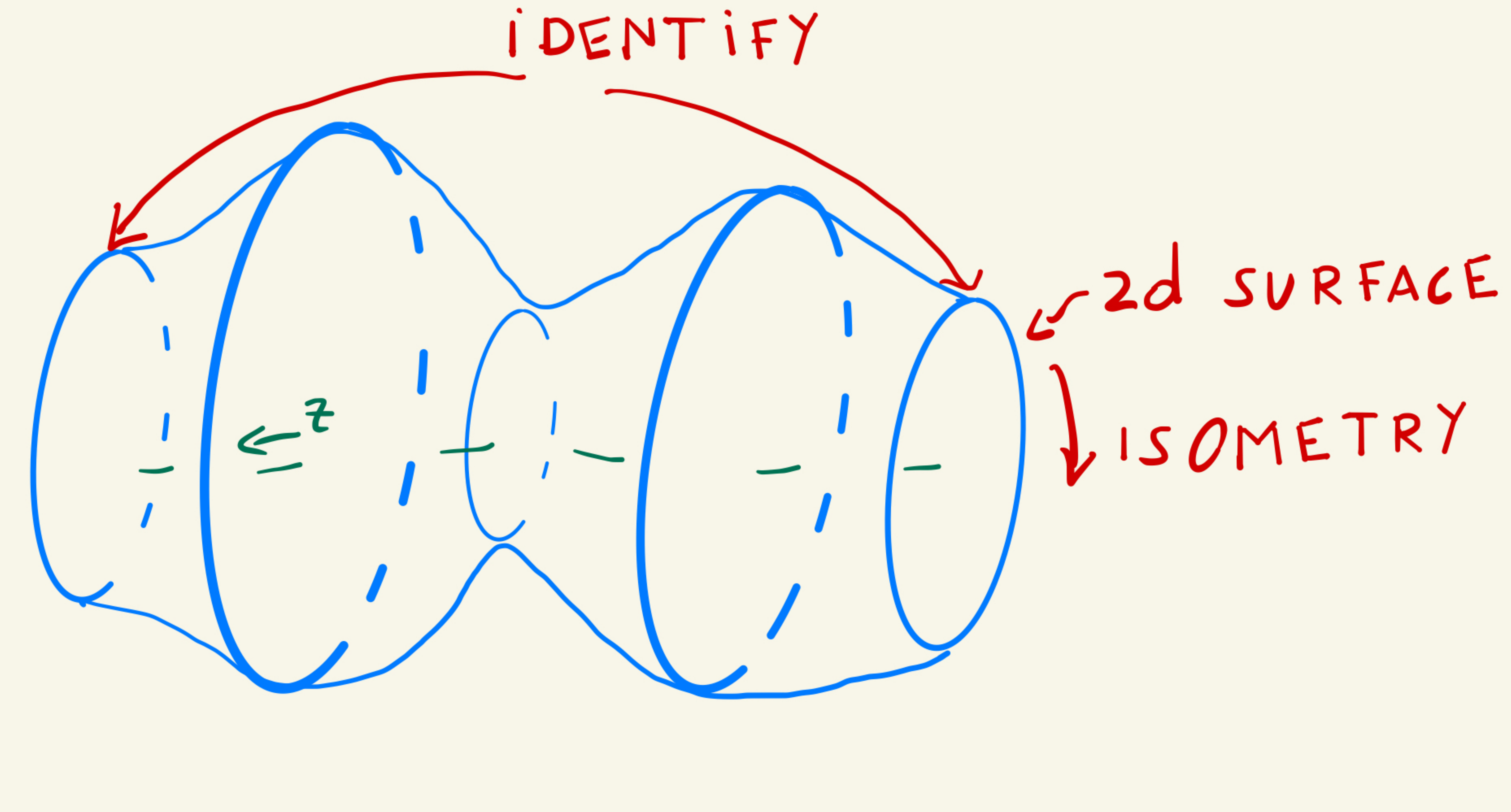}
\end{center}
\caption{\small
{A pictorial representation of an example of the geometry of the spatial slices allowed by our assumptions.} \label{fig:example_manifold} }
\end{figure}

\section{Notations and statement of main results}   

\paragraph{Notation and conventions.} The Riemann tensor is defined through $(\nabla_\mu\nabla_\nu - \nabla_\nu\nabla_\mu) \omega_\rho = R_{\mu\nu\rho}^{\quad\;\sigma} \omega_\sigma$, the Ricci tensor by $R_{\mu\nu} := R_{\mu\sigma\nu}^{\quad\;\sigma}$ (we also use the notation $\mathrm{Ric}(a,b)$, with $a,b$ being two vectors), the Ricci scalar {(also known as scalar curvature)} by $R := R_\mu^{\;\mu}$.

A time slice ${\cal M}_{\lambda}$ has an induced {Riemannian metric $g_{\mu\nu}$, and we can write  $g^{(4)}_{\mu\nu}=g_{\mu\nu}-n_\mu n_\nu$, where $g^{(4)}_{\mu\nu}$ is the spacetime metric (we use the mostly-plus convention) and  $n^\mu$} is orthonormal to ${\cal M}_\lambda$, $n_{\mu}n^{\mu} = -1$, and future-directed. The extrinsic curvature {(also known as second fundamental form)} of these slices is defined as $K_{\mu\nu} := g_\mu^{\;\alpha} \nabla_\alpha n_\nu$, satisfying $n^{\mu} K_{\mu \nu} = 0$ and with trace (also known as mean curvature) $K := g^{\mu\nu}K_{\mu\nu} = g^{(4)}{}^{\mu \nu} K_{\mu \nu}$, and traceless part $\sigma_{\mu \nu} := K_{\mu \nu} - {1 \over 3} K h_{\mu \nu}$ (with our sign convention $K>0$ corresponds to expansion). We also define $\sigma^2 := \sigma_{\mu\nu} \sigma^{\mu\nu}$; notice that $\sigma^2 \geq 0$, since $\sigma_{\mu\nu}$ is a tensor projected on the spatial hypersurfaces.  The Ricci tensor and Ricci scalar (scalar curvature) associated with the induced metric $g_{\mu\nu}$ on the $3$-dimensional slices are denoted, respectively, by  $^{(3)}R_{\mu\nu}$ and $^{(3)}R$.

Similarly, {each 2-dimensional symmetric orbit (or covering image of a symmetric orbit) $\Sigma$ (see Section \ref{symmetry_assumption_sec})} within ${\cal M}_\lambda$ has induced metric {$h_{\mu\nu}$ satisfying $g_{\mu\nu}=h_{\mu\nu}+t_\mu t_\nu$, where $t^\mu$ is orthogonal to $\Sigma$ and to $n_{\mu}$ and $t_{\mu}t^{\mu} = 1$}. The extrinsic curvature (second fundamental form) of this slice within ${\cal M}_\lambda$  is defined as $A_{\mu\nu} := h_\mu^{\;\alpha} \nabla_\alpha t_\nu$, satisfying $t^{\mu}A_{\mu \nu} = 0$ and with trace (mean curvature) $H := h^{\mu\nu}A_{\mu\nu}$. The Ricci tensor and Ricci scalar (scalar curvature) associated with the induced metric $h_{\mu\nu}$ on a $2$-dimensional slices are denoted by $^{(2)}R_{\mu\nu}$ and $^{(2)}R$ respectively.

We denote by the capital or lower case letters  $C_i$, and $D_i$, with $i=1,2,3,\ldots$, non-negative constants that depend only on the intrinsic and extrinsic properties of the initial 3-manifold of the flow: $M_{0}$.  We refer to such constant as universal.

\paragraph{Statement of main result.} We can now state the main theorem of this paper, the proof of which is spread in the following sections.

\begin{theorem}\label{main_thm_intro}
Let $M^{(3+1)}$ be a spacetime satisfying assumptions $(A)-(D)$ of Section \ref{general_ass_sec}, in addition to  the symmetry assumptions of Section  \ref{symmetry_assumption_sec}. If the orbit surfaces are spheres, one needs to further assume that the minimal area of an orbit surface in $M_0$ satisfies 
\begin{equation}
S_{\min} \geq S_{\mathrm{lower}},
\end{equation}  
where $S_{\mathrm{lower}}$ depends only on $\max_{x\in M_0} K$ and $K_{\Lambda}$  (see \eqref{eq:Slower}). Then there exists some  ${ 0\leq} \tilde{\lambda}<\infty$ and universal constants {$0< {d_1}<\infty$, $0\leq {d_2,d_3,d_4,d_5}<\infty$} such that
\begin{enumerate}[I.]
\item\label{foli_item} (The flow probes the entire future) The $\{\mathcal{M}_{\lambda}\}_{\lambda \geq 0}$ foliate $M^{(3+1)}\cap \{t\geq 0\}$.
\item\label{geocomp_item} {(Geodesic completeness and lack of singularities) $M^{(3+1)}\cap \{t\geq 0\}$ is future complete for timelike and null geodesics. There are no crushing singularities.} {
\item\label{flat_item} (Flatness of slices) For every  $\lambda \geq \tilde{\lambda}$ and every $p\in \mathcal{M}_{\lambda}$, the ball of radius $\frac{{ d_1}}{K_{\Lambda}}e^{\frac{1}{{12}}K_{\Lambda}^2(\lambda-\tilde{\lambda})}$ around $p$ in $\mathcal{M}_{\lambda}$ is $(1+ \frac{1}{10}e^{-\frac{1}{12}K_{\Lambda}^2(\lambda-\tilde{\lambda})})$-bi-Lipschitz equivalent to a Euclidean ball. In fact,}
\item\label{FLRW_item} (FLRW-expansion of slices) Taking any flow time $\lambda_0 \geq \tilde{\lambda}$, for any $\lambda>\lambda_0$ we can define the FLRW-expanding comparison metric on $\mathcal{M}_{\lambda}$
\begin{equation}
\mathbf{g}=e^{\frac{2}{3}K_{\Lambda}^2(\lambda-\lambda_0)} g_{\lambda_0}
\end{equation}
where the point identification is done by the MCF. Then 
\begin{equation}
||g(\lambda)-\mathbf{g}(\lambda)||_{g(\lambda)} \leq 2||g(\lambda)-\mathbf{g}(\lambda)||_{\mathbf{g}(\lambda)}\leq { d_2} e^{-\frac{1}{6}K_{\Lambda}^2(\lambda_0-\tilde{\lambda})}.
\end{equation}
\item\label{len_conv_item} (Length convergence to de Sitter of timelike and null curves) Let $\gamma:[0,a]\rightarrow M^{(3+1)}$ be a future-pointing timelike or null curve, with $\lambda(\gamma(0)) \geq \tilde{\lambda}$. Setting $\lambda_0=\lambda(\gamma(0))$ and $\lambda_a=\lambda(\gamma(a))$, we have
\begin{equation}
\left|L^{g^{(3+1)}}[\gamma]-L^{{\mathbf{g}_{\mathrm{dS}}^{(4)}}}[\gamma]\right| \leq \frac{{ d_3}}{K_{\Lambda}} e^{-\frac{1}{18}K_{\Lambda}^2(\lambda_0-\tilde{\lambda})}+ { d_3}K_\Lambda e^{-\frac{1}{{24}}K_{\Lambda}^2(\lambda_0-\tilde{\lambda})} \left(\lambda_a-\lambda_0\right),
\end{equation}  
where {$\mathbf{g}_{\mathrm{dS}}^{(4)}$ is a de Sitter metric 
\begin{equation}
\mathbf{g}_{\mathrm{dS}}^{(4)}=-K_{\Lambda}^2d\lambda^2+e^{\frac{2}{3}K_{\Lambda}^2(\lambda-\lambda_0)} g_{\mathrm{Euc}},
\end{equation} 
with $g_{\mathrm{Euc}}$ is some Euclidean metric on (part of) $\mathcal{M}_{\lambda_0}$, and where the point identification is done by the MCF. }
\item\label{matter_item} ($\mathrm{L}^1$ Dilution of matter) While $\frac{1}{2} \leq \frac{\mathrm{Vol}(\mathcal{M}_\lambda)}{\mathrm{Vol}(\mathcal{M}_ {\tilde{\lambda}})e^{K_{\Lambda}^2(\lambda-\tilde{\lambda})}} \leq 2$, 
\begin{equation}
\int_{\mathcal{M}_{\lambda}}\|T\|\;d\mathrm{Vol} \leq \frac{{ d_4}}{K_{\Lambda}}e^{\frac{1}{3}K_{\Lambda}^2(\lambda-\tilde{\lambda})}\ ,
\end{equation}
which is a slower rate than the volume by $e^{-\frac{2}{3}K_{\Lambda}^2\lambda}$. { The norm in this statement is the maximum of the components of the stress tensor, $T$, in an orthonormal frame whose time direction is orthogonal to the surfaces of mean curvature flow, ${\cal M}_\lambda$, or equivalently, the norm of $T$ with respect to a Riemannian metric associated to the Lorentzian metric $g^{(4)}$ via the flow.} Furthermore, letting $\gamma_{\lambda}$ be a geodesic in $\mathcal{M}_{\lambda}$, orthogonal to the orbit surfaces, and passing through each orbit surface once, we get that $\frac{1}{2} \leq \frac{L(\gamma_{\lambda})}{L(\gamma_{\tilde{\lambda}})e^{\frac{1}{3}K_{\Lambda}^2(\lambda-\tilde{\lambda})}} \leq 2$, but
\begin{equation}
\int_{\gamma_{\lambda}}\|T\|\;d\ell \leq { d_5}K_{\Lambda}e^{-\frac{1}{3}K_{\Lambda}^2(\lambda-\tilde{\lambda})}\ ,
\end{equation}
which is a slower rate than the length $L$ by $e^{-\frac{2}{3}K_{\Lambda}^2\lambda}$.
\end{enumerate}
\end{theorem} 

The proof of Theorem \ref{main_thm_intro} (and more) occupies the upcoming four sections. In Section \ref{sec:asymptotics_geometrical} we study the asymptotic behavior of the volume, length of the {transverse} geodesics $\gamma_{\lambda}$ of Theorem \ref{main_thm_intro},  and the minimal area of orbit surfaces. \ref{flat_item} of Theorem \ref{main_thm_intro} is proved in Section \ref{prox_glob_flat_sec}. \ref{FLRW_item} of Theorem \ref{main_thm_intro} is proved in Section \ref{uncond_point}. \ref{foli_item},\ref{geocomp_item} and \ref{len_conv_item} are proved in Section \ref{space_time_sec}.  \ref{matter_item} is proved in Section \ref{dilution_sec}. Section \ref{equiv_de_sit_sec} includes a discussion of why the results, as summarized in Theorem \ref{main_thm_intro}, imply asymptotic physical equivalence to de Sitter space.

\section{Asymptotic behavior of minimal surfaces, transverse length and spatial volume \label{sec:asymptotics_geometrical}}\label{ass_behave_sec}

\paragraph{Easy consequences.} {Contracting the Gauss equation for space-like hypersurface in a Lorentzian manifold \footnote{{Note that the second fundamental form term appears with an opposite sign compared to the Riemannian Gauss equation.}} twice, we get
\begin{equation}
R+2\mathrm{Ric}(n,n)={^{(3)}\!R}-K_{\mu\nu}K^{\mu\nu}+K^2={^{(3)}\!R}-\sigma^2+\frac{2}{3}K^2\ ,
\end{equation}
(see for instance \cite{Wald84}, eq.~(E.2.27)) which, combined with \eqref{play_ein1} and \eqref{K_lambda_def} gives
\begin{equation}
{^{(3)}\!R}-\sigma^2+\frac{2}{3}K^2=16\pi G_{N}(T(n,n)+\Lambda)=16\pi G_{N}T(n,n)+\frac{2}{3}K_{\Lambda}^2\ ,
\end{equation}
or in coordinate form:
\be
\label{eq0}
{^{(3)}\!R} + \frac{2}{3} K^2 - \sigma^2 = \frac23 K_\Lambda^2 + 16\pi G_N T_{\mu\nu} n^\mu n^\nu \;.
\ee
By WEC, we have
\begin{equation}\label{dom_en}
{^{(3)}\!R}+\frac{2}{3}K^2-\sigma^2\geq \frac{2}{3}K_{\Lambda}^2\ ,
\end{equation}
so by Theorem~\ref{th:boundonmax}, we have the following pointwise bound on ${^{(3)}\!R}$:
\begin{equation}\label{pointwise_scalar_low}
{^{(3)}\!R}\geq -C_2 K_\Lambda^2 e^{-\frac{2}{3}K_{\Lambda}^2\lambda} \ ,
\end{equation}
where $C_2=\frac{2}{3}C_1(2+C_1)$ and $C_1=\max(K_m(0)/K_\Lambda-1, 0)$, as in Theorem \ref{th:boundonmax}.

\paragraph{Growth of geometric quantities.} We are now going to establish the growth of some geometric quantities defined  {along the mean curvature flow hypersurfaces ${\cal M}_\lambda$}. Fix some time $\lambda\geq 0$, ${\cal M}_{\lambda}$,  and consider the foliation of ${\cal M}_{\lambda}$ by the orbits of $G$ (or more generally, by the {projections} of the orbits of $\tilde{G}$).  By our two-sidedness assumption, there exists a global unit normal vector $E$ to this foliation.  Let $z$ be the parameter along the flow lines of $E$, thus it is a signed distance function; and, due to the isometries of $G$ (or $\tilde{G}$),  the metric on ${\cal M}_{\lambda}$ has the warped product form 
\begin{equation}\label{eq:spatial_metric}
g=dz^2+h_z\ ,
\end{equation}
where $h_z$ is a two-dimensional metric of constant curvature. By passing to a double cover, we can assume without loss of generality that the orbit surfaces $\Sigma$ are orientable. Thus, each such orbit is a two-dimensional orientable surface, with Euler characteristic $\chi = 2,0,-2,-4, \ldots$.

 We will start {by proving} that {the minimal area of a surface orbit contained in $\mathcal{M}_{\lambda}$, which we denote by $S_{\min}(\lambda)$, grows as two-dimensional} spatial slices of de Sitter space in the FLRW slicing. In order to study the time evolution of $S_{\min}(\lambda)$, we would like to find a differential equation for $S_{\min}(\lambda)$ and solve for it. However, since the area of the minimal surface can be {non-differentiable}   as the flow evolves, it is unclear that this can be done. Therefore, we first need to show that $S_{\min}(\lambda)$ has well defined derivatives almost everywhere and that the fundamental theorem of calculus applies to them.  We do this by proving that they are Lipschitz. This is true because of the following {standard} lemma which applies to all minimizers:

\begin{lemma}[Hamilton's trick (c.f \cite{Manta} Lemma 2.1.3)]\label{Ham_trick}
Let $f:K\times [a,b]\rightarrow \mathbb{R}$ be a smooth function {with $K$ being compact}, and set $g$ to be the minimizer of $f$ on $K$:
\begin{equation}
g(t)=\min_{x\in K}f(x,t).
\end{equation}
Then $g$ is a Lipschitz function, and thus, differentiable almost everywhere and obeying the fundamental theorem of calculus. Moreover, if $t_0$ is a point of differentiability of $g$, and if $x_0$ is such that $f(x_0,t_0)=g(t_0)$ then 
\begin{equation}
g'(t_0)=\left.\frac{\partial f}{\partial t}\right|_{(x_0,t_0)}\ .
\end{equation}
\end{lemma}
\begin{proof}
First, we show that $g$ is Lipschitz. For every $t$, let $x_t$ be a point  such that $g(t)=f(x_t,t)$. Then for every $t,s\in[a,b]$, we have
\begin{equation}
g(s)-g(t) \leq f(x_t,s)-f(x_t,t)=\int_t^s dt' \; \frac{\partial f}{\partial t}(x_t,t')\leq C_3 |t-s|, 
\end{equation}
where $C_3=\max_{(x_1,t_1)\in K\times [a,b]}\frac{\partial f}{\partial t}(x_1,t_1)$.
Similarly, $g(t)-g(s) \leq C_3|t-s|$, so $g$ is indeed Lipschitz, hence differentiable almost everywhere  and obeying the fundamental theorem of calculus. 

Let $t_0$ be a point of differentiability of $g$. In particular 
\begin{equation}
\lim_{t\searrow t_0} \frac{g(t)-g(t_0)}{t-t_0}=g'(t_0)=\lim_{t\nearrow t_0} \frac{g(t_0)-g(t)}{t_0-t} \ .
\end{equation}
For $t<t_0$, we have that
\begin{equation}
g(t_0)-g(t)\geq f(x_{t_0},t_0)-f(x_{t_0},t),
\end{equation}
so dividing both sides by $t_0-t$ and taking the limit as $t\nearrow t_0$, we obtain
\begin{equation}
g'(t_0) \geq \left.\frac{\partial f}{\partial t}\right|_{(x_0,t_0)}.
\end{equation}
Similarly, for $t>t_0$, we have 
\begin{equation}
g(t)-g(t_0)\leq f(x_{t_0},t)-f(x_{t_0},t_0) \ ,
\end{equation}
so dividing both sides by $t-t_0$ and taking the limit as $t\searrow t_0$ we also get
\begin{equation}
g'(t_0) \leq \left.\frac{\partial f}{\partial t}\right|_{(x_0,t_0)}.
\end{equation} 
This proves the claim.
\end{proof}

We can now prove the following theorem on the area growth of { the minimal orbit surface}:

\begin{theorem}\label{surface_growth_lemma}
Denote by $S_{\min}(\lambda)$  the minimal  area of a $z$-cross section and $\chi$ its Euler characteristic. Then if either $\chi\leq0$, or, if {$\chi=2$}, if also $S_{\min}(0)\geq S_{\rm lower}$, then there exists $\lambda_{ 0,1}$ such that for all $\lambda_{ 0,1}<\lambda_1<\lambda_2$:
\begin{equation}
\frac{1}{2} \leq \frac{S_{\min}(\lambda_2)}{S_{\min}(\lambda_1)e^{\frac{2}{3}K_{\Lambda}^2(\lambda_2-\lambda_1)}} \leq 2\ ,
\end{equation}
where 
\be\label{eq:Slower}
S_{\rm lower}= \frac{8 \pi}{K_\Lambda^2} \frac{1}{\left(\sqrt{(1+C_1)^2+2/9}-(1+C_1)\right)^2}\, (9C_4)^{3 C_4-4/3}\ ,
\ee
with $C_1=\max(K_m(0)/K_\Lambda-1, 0)$ as in Theorem \ref{th:boundonmax} and $C_4= \frac{4}{\sqrt{3}}\sqrt{C_1
   (C_1+2)} (C_1+1)+\frac{4}{3} C_1
   (C_1+2)$.
\end{theorem}

\begin{proof}
Recall that the function $S_{\min}(\lambda)$ are (locally) Lipschitz functions, and hence differentiable almost everywhere. Also, recall that at differentiable times $\lambda$ for $S_{\min}$  the derivative will be identical to the derivative of the area of the section where the minimum  is obtained (see Lemma \ref{Ham_trick}).

By the Riccati equation (primes indicate derivatives w.r.t.~$z$)
\be\label{Rica_eq}
H'+A_{\mu\nu}A^{\mu\nu}=-{}^{(3)}\!R_{zz}\ .
\ee
Now, the traced Gauss equations imply
\begin{equation}
{}^{(3)}\!R={}^{(2)}\!R+2{}^{(3)}\!R_{zz}+A_{\mu\nu}A^{\mu\nu}-H^2\ ,
\end{equation}
so 
\be\label{Gauss_eq}
-{}^{(3)}\!R_{zz}=\frac{{}^{(2)}\!R-{}^{(3)}\!R+A_{\mu\nu}A^{\mu\nu}-H^2}{2}\;.
\ee
Combining \eqref{Rica_eq} and \eqref{Gauss_eq}, we obtain

\begin{equation}\label{basic_eq}
H'+A_{\mu\nu}A^{\mu\nu}=-{}^{(3)}\!R_{zz}=\frac{{}^{(2)}\!R-{}^{(3)}\!R+A_{\mu\nu}A^{\mu\nu}-H^2}{2} \ .
\end{equation}

Consider a $z$ slice with minimal area. On this $z$ slice we have $H=0$ and $H'\geq 0$, so \eqref{basic_eq} gives
\begin{equation}
{}^{(3)}\!R = {}^{(2)}\!R-A_{\mu\nu}A^{\mu\nu}-{H^2}-2H' \leq {}^{(2)}\!R\ ,
\end{equation}  
which, combined with \eqref{dom_en} gives 
\begin{equation}\label{eq:temp}
\frac{2}{3}K_{\Lambda}^2-\frac{2}{3}K^2+\sigma^2\leq {}^{(2)}\!R\ ,
\end{equation}
on such a slice. 
Notice that by our isometries, if $S(z,\lambda)$ is the area of a fixed $z$ surface at time $\lambda$:
\bea\label{eq:R2ineq}
{}^{(2)}\!R(z,\lambda){\leq} \frac{4\pi\chi_0}{S(z,\lambda)}\leq \frac{4\pi\chi_0}{S_{\rm min}(\lambda)}\ ,
\eea
where $\chi_0$ is $2$ if $\tilde{\Sigma}$ is the sphere and $0$ otherwise. 
Eq.~(\ref{eq:temp}) and  Theorem~\ref{th:boundonmax} imply that, considering either cases in which $K\leq K_\Lambda$ or $K> K_\Lambda$, we have
\begin{equation}\label{abs_est}
\frac{2}{3}\left|K_{\Lambda}^2-K^2\right|+\sigma^2\leq {}^{(2)}\!R+2C_2K_\Lambda^2 e^{-\frac{2}{3}K_{\Lambda}^2\lambda}\leq \frac{4\pi\chi_0}{S_{\rm min}}+2C_2K_\Lambda^2 e^{-\frac{2}{3}K_{\Lambda}^2\lambda} \ .
\end{equation}

The evolution equation for the metric under MCF {(see \cite[Prop. 3.1]{EH})} is{  
\begin{equation}
\frac{dg_{ij}}{d\lambda}=2KK_{ij}=\frac{2}{3}K^2 g_{ij}+2K\sigma_{ij}=\frac{2}{3}K_{\Lambda}^2g_{ij}+E_{S,ij}\ ,
\end{equation} 
where $E_{S,ij}=\frac{2}{3}(K^2-K_\Lambda^2)g_{ij}+2K \sigma_{ij}$. 
We want now to bound this equation using the previous inequalities. Note that, using a more abstract notation, we can write, with no summation over repeated indexes,
\be
\frac{2}{3}|(K^2-K_\Lambda^2)g_{ij}|=\frac{2}{3}|K^2-K_\Lambda^2|\left|g\left(\frac{\partial_i}{||\partial_i||},\frac{\partial_j}{||\partial_j||}\right)\right|\sqrt{g_{ii}}\sqrt{g_{jj}}\leq \frac{2}{3}|K^2-K_\Lambda^2|\sqrt{g_{ii}}\sqrt{g_{jj}}\  ,  
\ee
where in the last step we used that, for two unit vectors $\hat n_1$ and $\hat n_2$, 
\be
g(\hat n_1,\hat n_2)\leq \sqrt{g(\hat n_1,\hat n_1)}\sqrt{g(\hat n_2,\hat n_2)}=1\ .
\ee
Similarly, we can write 
\begin{equation}
\sigma_{ij}^2=\sigma\left(\frac{\partial_i}{||\partial_i||},\frac{\partial_j}{||\partial_j||}\right)^2g_{ii}g_{jj}\ .
\end{equation}
So, putting this together with the inequalities \eqref{abs_est} and Theorem~\ref{th:boundonmax}, we get 
\begin{align}
|E_{S,ij}|\leq  \left[\left(\frac{4\pi\chi_0}{S_{\rm min}}+2C_2 K_\Lambda^2 e^{-\frac{2}{3}K_{\Lambda}^2\lambda}\right)+2K_\Lambda(1+C_1) \left(\sqrt{\frac{4\pi\chi_0}{S_{\rm min}}}+\sqrt{2C_2} K_\Lambda e^{-\frac{1}{3}K_{\Lambda}^2\lambda}\right)\right]\sqrt{g_{ii}}\sqrt{g_{jj}}\ ,
\end{align}
where we used that $\sqrt{a+b}\leq \sqrt{a}+\sqrt{b}$.
Choosing the coordinates on the surface, $x_1,x_2$, to be orthonormal at the point at time $\lambda$, the area form of the cross section surface at that point, at varying times, is given by $\sqrt{\mathrm{det}^{12}g_{ij}}dx^1dx^2$,  and we obtain: 
\begin{equation}
\frac{1}{\sqrt{\mathrm{det}^{12}g_{ij}}}\frac{d}{d\lambda}\sqrt{\mathrm{det}^{12}g_{ij}}|_{\lambda}=\frac{1}{2}\mathrm{tr}^{12}\left(\frac{dg_{ij}}{d\lambda}|_{\lambda}\right)=\frac{2}{3}K_{\Lambda}^2+ { \frac12} E_S\ ,
\end{equation}
where $E_S$ satisfies 
\begin{equation}
|E_S|\leq { 2}\left(\left[\frac{4\pi\chi_0}{S_{\rm min}}+2C_2 K_\Lambda^2 e^{-\frac{2}{3}K_{\Lambda}^2\lambda}\right]+2K_\Lambda (1+C_1) \left[\sqrt{\frac{4\pi\chi_0}{S_{\rm min}}}+\sqrt{2C_2} K_\Lambda e^{-\frac{1}{3}K_{\Lambda}^2\lambda}\right]\right)\ .
\end{equation}
Thus, at such a slice 
\begin{equation}
\frac{d}{d\lambda}dS=\left(\frac{2}{3}K_{\Lambda}^2+ { \frac12} E_S\right)dS\ .
\end{equation}
Integrating over that slice and using Lemma \ref{Ham_trick}, we see that at every $\lambda$ where $S_{\min}(\lambda)$ is differentiable, 
\begin{align}
&K_{\Lambda}^2\left(\frac{2}{3}-C_4 e^{-\frac{1}{3}K_{\Lambda}^2\lambda}\right)S_{\min}(\lambda)-4\pi\chi_0-2(1+C_1)K_\Lambda\sqrt{4\pi\chi_0 S_{\min}(\lambda)} \leq \frac{d}{d\lambda}S_{\min}(\lambda)\\ \nonumber
&\qquad\qquad\qquad\leq K_{\Lambda}^2\left(\frac{2}{3}+C_4\, e^{-\frac{1}{3}K_{\Lambda}^2\lambda}\right)S_{\min}(\lambda)+4\pi\chi_0+2(1+C_1)K_\Lambda\sqrt{4\pi\chi_0 S_{\min}(\lambda)}\ ,
\end{align} 
where $C_4=2(1+C_1)\sqrt{2C_2}+2C_2$.
Thus, at such point of differentiability,
\begin{equation}\label{log_growth}
\left|\frac{d}{d\lambda}\log(S_{\min})-\frac{2}{3}K_{\Lambda}^2\right| \leq \frac{4\pi\chi_0}{S_{\min}}+2(1+C_1)K_\Lambda\frac{\sqrt{4\pi\chi_0}}{\sqrt{S_{\min}}}+K_{\Lambda}^2 C_4\, e^{-\frac{1}{3}K_{\Lambda}^2\lambda} \ .
\end{equation}
Now, if $\chi\leq 0$, i.e.~$\chi_0 =0$
\begin{equation}\label{log_growthneg}
\left|\frac{d}{d\lambda}\log(S_{\min})-\frac{2}{3}K_{\Lambda}^2\right| \leq K_{\Lambda}^2 C_4\, e^{-\frac{1}{3}K_{\Lambda}^2\lambda} 
\end{equation}{
and there exists a time $\lambda_{0,1}'$ such that 
\begin{equation}\label{eq:C4bound}
K_{\Lambda}^2 C_4\, e^{-\frac{1}{3}K_{\Lambda}^2\lambda_{0,1}'}=\frac{1}{9}K_{\Lambda}^2\ .
\end{equation}
If instead $\chi=2$, note that as long as $S_{\min}\geq \hat S:= \frac{8 \pi}{K_\Lambda^2} \frac{1}{\left(\sqrt{(1+C_1)^2+2/9}-(1+C_1)\right)^2}$ (ensuring that the first two terms in the right hand side of \eqref{log_growth} contribute at most $\frac{2}{9}K_{\Lambda}^2$)
\begin{equation}
\frac{d}{d\lambda}\log(S_{\min}) \geq \left(\frac{4}{9}-C_4\right)K_{\Lambda}^2:=-C_4'K_{\Lambda}^2\ .
\end{equation}
We therefore get that  if $S_{\min}(0) \geq S_{\rm lower}$, where
\be
S_{\rm lower}:=\hat S\, e^{C_4'K_{\Lambda}^2\lambda_{0,1}'}=\hat S\,(9C_4)^{3C_4'}
\ee
 then $S_{\min}(\lambda) \geq  \hat S$ on $[0,\lambda_{0,1}']$ with  $\lambda_{0,1}'$ defined by (\ref{eq:C4bound}), and at $\lambda_{0,1}'$, 

\begin{equation}\label{eq:closedbound}
\frac{{8\pi}}{S_{\min}(\lambda_{0,1}')}+\frac{2 (1+C_1) {\sqrt{8\pi}}K_\Lambda}{\sqrt{S_{\min}(\lambda_{0,1}')}}+K_{\Lambda}^2 C_4\, e^{-\frac{1}{3}K_{\Lambda}^2\lambda_{0,1}'}\leq \frac{1}{3}K_{\Lambda}^2\ .
\end{equation}}

Assuming, for $\chi=2$, that $S_{\rm min}(0)>S_{\rm lower}$, we can integrate \eqref{log_growth} for any $\chi_0$,  and get for every $\lambda\geq \lambda_{0,1}'$ the non-optimal estimate 
\begin{equation}
S_{\min}(\lambda)\geq S_{\min}(\lambda_{0,1}')e^{\frac{1}{3}K_{\Lambda}^2(\lambda-\lambda_{0,1}')}\ .
\end{equation}
Substituting back to \eqref{log_growth}, we obtain   
\begin{align}\label{log_growth_2}
&\left|\frac{d}{d\lambda}\log(S_{\min})-\frac{2}{3}K_{\Lambda}^2\right| \leq\\\nonumber
&\qquad \leq \frac{4\pi\chi_0}{S_{\min}(\lambda_{0,1}')}e^{-\frac{1}{3}K_{\Lambda}^2(\lambda-\lambda_{0,1}')}+2(1+C_1)K_\Lambda\frac{\sqrt{4\pi\chi_0}}{\sqrt{S_{\min}(\lambda_{0,1}')}}e^{-\frac{1}{6}K_{\Lambda}^2(\lambda-\lambda_{0,1}')}+K_\Lambda^2C_4\, e^{-\frac{{1}}{3}K_{\Lambda}^2\lambda}\ .
\end{align}
Now, let $\lambda_{0,1}\geq \lambda_{0,1}'$ be such that 
\begin{align}\label{bound_int_inf}\nonumber
&\int_{\lambda_{0,1}}^{\infty}d\lambda\;\left(\frac{4\pi\chi_0}{S_{\min}(\lambda_{0,1}')}e^{-\frac{1}{3}K_{\Lambda}^2(\lambda-\lambda_{0,1}')}+2(1+C_1)K_\Lambda \frac{\sqrt{4\pi\chi_0}}{\sqrt{S_{\min}(\lambda_{0,1}')}}e^{-\frac{1}{6}K_{\Lambda}^2(\lambda-\lambda_{0,1}')}+K_\Lambda^2 C_4\, e^{-\frac{{1}}{3}K_{\Lambda}^2\lambda}\right)\\ 
&\qquad \leq \log 2\ .
\end{align}
Then integrating \eqref{log_growth_2} from $\lambda_1$ to $\lambda_2$, where $\lambda_{0,1} \leq \lambda_1<\lambda_2$, and using \eqref{bound_int_inf}, we obtain
\begin{equation}
\left| \log\left(\frac{S_{\min}(\lambda_2)}{S_{\min}(\lambda_1)}\right)-\frac{2}{3}K_{\Lambda}^2(\lambda_2-\lambda_1)\right| \leq \log 2\ ,
\end{equation}
so exponentiating both sides yields the desired result.}
\end{proof}
{Notice that the additional requirement in the case of the sphere depends exponentially on the initial conditions. This is different from what happens in the case of complete homogeneity where, for Bianchi-IX universes, one has to impose a lower bound on $^{(3)} R$ \cite{Wald:1983ky}. This bound however does not depend exponentially on the initial conditions. }

By the form of the metric in (\ref{eq:spatial_metric}), it is straightforward to check that if a geodesic is at a point tangent to the vector $E$, it is tangent to $E$ everywhere.  Denote therefore by $L(\lambda)$ the length, at time $\lambda$, of any geodesic $\gamma$  that is parallel to the $z$-direction, from an initial slice to itself. Additionally, denote by $V(\lambda)$ the volume of ${\cal M}_\lambda$ at time $\lambda$.

\begin{theorem}\label{length_growth_lemma}
Under the conditions of Theorem \ref{surface_growth_lemma}, for every $\delta >0$, there exists $\lambda_{0,2}{\geq}\lambda_{0,1}$ such that for every $\lambda>\lambda_{0,2}$ 
\begin{equation}
(1+\delta)^{-1} \leq  \frac{L(\lambda)}{L(\lambda_{0,2})e^{\frac{1}{3}K_{\Lambda}^2(\lambda-\lambda_{0,2})}}\leq 1+\delta\ .
\end{equation}
and
\begin{equation}\label{vol_growth_temp}
(1+\delta)^{-1} \leq\frac{V(\lambda)}{V(\lambda_{0,2})e^{K_{\Lambda}^2(\lambda-\lambda_{0,2})}}\leq 1+\delta\ .
\end{equation}
\end{theorem}

\begin{proof}
Re-arranging \eqref{basic_eq}, we obtain 
\begin{equation}\label{basic_eq_2}
{}^{(3)}\!R=-A_{\mu\nu}A^{\mu\nu}-H^2+{}^{(2)}\!R-2H'\ .
\end{equation}
Let us integrate \eqref{basic_eq_2} along all the $z$-directed geodesic. By the periodicity, the term in $H'$ does not contribute. Therefore, using (\ref{eq:R2ineq}), we obtain

\begin{equation}\label{first_L}
\int_0^{L(\lambda)} dz\;{}^{(3)}\!R \leq  \int_{0}^{L(\lambda)}dz\; {}^{(2)}\!R\leq   \int_{0}^{L(\lambda)} dz\;\frac{2\cdot 4\pi\chi_0}{S_{\min}(\lambda_{0,1})e^{\frac{2}{3}K_{\Lambda}^2(\lambda-\lambda_{0,1})}}=K_\Lambda^2 C_5e^{-\frac{2}{3}K_{\Lambda}^2\lambda}L(\lambda)
\end{equation}
where 
\be
C_5:=\frac{8\pi\chi_0}{K_\Lambda^2 S_{\min}(\lambda_{0,1})e^{-\frac{2}{3}K_{\Lambda}^2\lambda_{0,1}}} \ ,
\ee
and where we used Theorem~\ref{surface_growth_lemma}, since {$\lambda>\lambda_{0,1}$, given that for this Theorem we are assuming $\lambda> \lambda_{0,2}{\geq}\lambda_{0,1}$}.
In light of \eqref{dom_en}, we therefore have that
\begin{equation}\label{second_L}
\int_0^{L(\lambda)}dz\;\left( \frac{2}{3}\left(K_{\Lambda}^2-K^2\right) +\sigma^2 \right)\leq K_\Lambda^2 C_5 \; e^{-\frac{2}{3}K_{\Lambda}^2\lambda}L(\lambda)\ .
\end{equation}
This implies that, using Theorem~\ref{th:boundonmax}:
\be\label{eq:sigmabound1}
\int_0^{L(\lambda)}dz\;\sigma^2 \leq K_\Lambda^2 C_5 e^{-\frac{2}{3}K_{\Lambda}^2\lambda}L(\lambda)+\int_0^{L(\lambda)}dz\; \frac{2}{3}\left(K^2-K_\Lambda^2\right)\leq K_\Lambda^2\left(C_5+C_2\right) e^{-\frac{2}{3}K_{\Lambda}^2\lambda}L(\lambda)\ .
\ee
Using again Theorem~\ref{th:boundonmax}, we therefore get that
\begin{equation}\label{quant_eq}
\int_{0}^{L(\lambda)} dz\left(\frac{2}{3}|K_{\Lambda}^2-K^2|+ \sigma^2\right) \leq K_\Lambda^2 C_6 e^{-\frac{2}{3}K_{\Lambda}^2\lambda}L(\lambda)\ ,
\end{equation}
where $C_6=C_5+2C_2$.
{Computing
\begin{align}
L'(\lambda)&=\int_0^{L(\lambda)}dz\; KK_{zz}=\int_0^{L(\lambda)}dz\;\left(\frac{K^2}{3}+K\sigma_{zz}\right)=\\ \nonumber
&=\frac{K_{\Lambda}^2}{3}L(\lambda)+\int_0^{L(\lambda)}dz\;\left(\frac{1}{3}(K^2-K_{\Lambda}^2)+K\sigma_{zz} \right)\ ,
\end{align}
we see that 
\begin{align}\label{ODE_L}
|L'(\lambda)- \frac{K_{\Lambda}^2}{3}L(\lambda)|\leq E_L(\lambda)\ ,
\end{align}
for $E_L(\lambda)$ satisfying
\begin{equation}\label{E_L_est}
|E_L(\lambda)|\leq  K_\Lambda^2\frac{C_6}{2}e^{-\frac{2}{3}K_{\Lambda}^2\lambda}L(\lambda)+ K_\Lambda^2\left(1+C_1e^{-\frac{2}{3}K_{\Lambda}^2\lambda}\right) (C_5+C_2)^{1/2} e^{-\frac{1}{3}K_{\Lambda}^2\lambda} L(\lambda)\ . 
\end{equation}  
Here, we have used \eqref{quant_eq} and,  for the term $K\sigma_{zz}$, we have used Theorem~\ref{th:boundonmax}, the Cauchy-Schwartz inequality for {$\int dz \;| \sigma_{zz}|$}} {and eq.~\eqref{eq:sigmabound1}}. {Integrating the ordinary differential inequalities {\eqref{ODE_L} (keeping in mind \eqref{E_L_est}), similarly to what done in Theorem \ref{surface_growth_lemma}, and defining   $\lambda_{0,2}'= \max(\lambda_{0,1},\bar\lambda_{0,2}')$, with $\bar\lambda_{0,2}'$ such that 
\begin{align}\label{bound_int_inf2}
&\int_{\bar\lambda_{0,2}'}^{\infty}d\lambda\;\left( K_\Lambda^2\frac{C_6}{2}e^{-\frac{2}{3}K_{\Lambda}^2\lambda}+ K_\Lambda^2\left(1+C_1e^{-\frac{2}{3}K_{\Lambda}^2\lambda}\right) (C_5+C_2)^{1/2} e^{-\frac{1}{3}K_{\Lambda}^2\lambda} \right) \leq \log (1+\delta)\ ,
\end{align}
we obtain
\begin{equation}
(1+\delta)^{-1} \leq  \frac{L(\lambda)}{L(\lambda_{0,2}')e^{\frac{1}{3}K_{\Lambda}^2(\lambda-\lambda_{0,2}')}}\leq 1+\delta\ .
\end{equation}
}

One can work quite similarly for the volume. Explicitly, we can write
\begin{align}
V'(\lambda)&=\int_{{\cal M}_\lambda}dV\; K^2=K_{\Lambda}^2V(\lambda)+\int_{{\cal M}_\lambda}dV\;\left(K^2-K_{\Lambda}^2\right)\ .
\end{align}
Notice that equations \eqref{first_L}, \eqref{second_L}, \eqref{eq:sigmabound1}, \eqref{quant_eq}, and their derivation hold verbatim if we replace $L(\lambda)$ with $V(\lambda)$, and integrals over $\gamma$ with integrals over ${\cal M}_{\lambda}$. In particular,
\begin{equation}\label{eq:volumeksigmaestimate}
\int_{\mathcal{M}_{\lambda}}\left(\frac{2}{3}|K_{\Lambda}^2-K^2|+ \sigma^2\right) \leq K_\Lambda^2 C_6 e^{-\frac{2}{3}K_{\Lambda}^2\lambda}V(\lambda)\ .
\end{equation}
We therefore see that
\begin{align}\label{ODE_V}
|V'(\lambda)- K_{\Lambda}^2V(\lambda)|\leq E_V(\lambda)\ ,
\end{align}
with
\begin{equation}\label{E_V_est}
|E_V(\lambda)|\leq  \frac{3}{2}K_\Lambda^2 C_6 e^{-\frac{2}{3}K_{\Lambda}^2\lambda}V(\lambda) \ . 
\end{equation} 
Integrating the ordinary differential inequalities \eqref{ODE_V}, and defining  $\lambda_{0,2}''= \max(\lambda_{0,1},\bar\lambda_{0,2}'')$, with $\bar\lambda_{0,2}''$ such that
\begin{align}\label{bound_int_inf2}
&\int_{\bar\lambda_{0,2}''}^{\infty}d\lambda\;\left( \frac{3}{2}K_\Lambda^2 C_6 e^{-\frac{2}{3}K_{\Lambda}^2\lambda} \right) \leq \log (1+\delta)\ ,
\end{align}
we obtain
\begin{equation}
(1+\delta)^{-1} \leq  \frac{V(\lambda)}{V(\lambda_{0,2}'')e^{K_{\Lambda}^2(\lambda-\lambda_{0,2}'')}}\leq 1+\delta\ .
\end{equation}
Choosing $\lambda_{0,2}=\max(\lambda_{0,2}',\lambda_{0,2}'')$, we obtain the desired result.}
\end{proof}

\begin{lemma}
There exists some $C_7$ such that
\begin{equation}\label{C_7_control}
\int_{{\cal M}_{\lambda}} dV \left(|K_{\Lambda}^2-K^2|+ \sigma^2\right) \leq \frac{1}{K_\Lambda} C_7\;e^{\frac{1}{3}K_{\Lambda}^2(\lambda-\lambda_{0,2})}
\end{equation} 
\end{lemma}
\begin{proof}
{The desired result is obtained by combining (\ref{eq:volumeksigmaestimate}) with  \eqref{vol_growth_temp}, with $C_7=\frac{3}{2}C_6 (1+\delta) K_\Lambda^3 V(\lambda_{0,2})e^{-\frac{2}{3}K_\Lambda^2\lambda_{0,2}}$.}
\end{proof}

\paragraph{Resetting of time:} Now, for ease of notation, let us re-define the initial time of the flow as to be $\lambda_{0,2}$, so from now on $\lambda_{0,2}=0$. Note that estimates \eqref{pointwise_scalar_low} and \eqref{eq:maxbound} still hold (in fact, with much better constants). \\

In particular, we have, for every $\lambda>0$

\begin{equation}\label{vol_est1}
\frac{1}{2}\leq\frac{V(\lambda)}{V(0)e^{K_{\Lambda}^2\lambda}}\leq 2\ ,
\end{equation}  
\begin{equation}\label{area_est1}
\frac{1}{2} \leq \frac{S_{\min}(\lambda)}{S_{\min}(0)e^{\frac{2}{3}K_{\Lambda}^2\lambda}}\leq 2\ ,
\end{equation}
and 
\begin{equation}\label{quant_est}
\int_{{\cal M}_{\lambda}} dV\;\left(|K_{\Lambda}^2-K^2|+ \sigma^2\right) \leq  \frac{1}{K_\Lambda}C_7\;e^{\frac{1}{3}K_{\Lambda}^2\lambda}\ .
\end{equation}

\section{Spatial closeness\label{sec: spatialprox}}
In this section we focus on the spatial part of the metric, {\it i.e.}~the induced metric on the {hypersurfaces $\mathcal{M}_{\lambda}$ at  fixed $\lambda$}. One can define a comparison metric
\begin{equation}\label{comp_metric}
\mathbf{g} := g(\lambda_0)e^{\frac{2}{3}K_{\Lambda}^2(\lambda-\lambda_0)}. 
\end{equation}
This corresponds to evolving in $\lambda$, starting from $\lambda_0$, the spatial metric of ${\cal M}_{\lambda_0}$, with {the} same rate as the flat slicing of de Sitter.  We are going to prove that the metric on the surfaces at constant $\lambda$ converges pointwise, for large $\lambda_0$, to this comparison metric. 
{At the end of this section, in \ref{prox_glob_flat_sec}, we will construct a \textit{genuinely}-flat spatial metric that expands in time as the flat slices of de Sitter space, which approximates $g$ over expanding balls.}

\subsection{\label{sec:zdep}Propagation of the metric along the level set}  

{In this section we are going to show that, as $\lambda$ becomes larger and larger, the spatial metric $g$ of ${\cal M}_\lambda$ becomes less and less dependent on the transverse direction~$z$.} The propagation of the metric along the level sets is given by the second fundamental form (extrinsic curvature):
\begin{equation}\label{metric_prop}
\mathcal{L}_{\partial_z} g_{\mu\nu}=2 A_{\mu\nu}\ .
\end{equation}
Now, using eq.~\eqref{basic_eq}, the pointwise bounds  ${}^{(2)}\!R$ given by (\ref{eq:R2ineq}) and \eqref{area_est1} and the one on ${}^{(3)}\!R$ in eq.~\eqref{pointwise_scalar_low}, we get 
\begin{equation}\label{less_eps}
H'+\frac{H^2+A_{\mu\nu}A^{\mu\nu}}{2} \leq \left( \frac{4\pi\chi_0 }{S_{\rm min}(0)}+ \frac{C_2 K_\Lambda^2}{2}\right)e^{-\frac{2}{3}K_{\Lambda}^2\lambda}:= C_8 { K_\Lambda^2}\,e^{-\frac{2}{3}K_{\Lambda}^2\lambda}\ .
\end{equation}
This implies the following pointwise bound of $|H|$:
\begin{claim}
\begin{equation}\label{H_bound}
|H|\leq \frac{2}{\sqrt{3}}\sqrt{C_8}{ K_\Lambda}e^{-\frac{1}{3}K_{\Lambda}^2\lambda}:=\eps_{\lambda}\ .
\end{equation}
\end{claim}
\begin{proof}
At the minimum and maximum points of $H$, $H'=0$ and \eqref{H_bound} follows from \eqref{less_eps} there. If  \eqref{H_bound} holds at the minimum and the maximum, it holds at any point.
\end{proof}

Integrating \eqref{less_eps}, and using the pointwise bound \eqref{H_bound}, we also get

\be
\int_{0}^{z}A_{\mu\nu}A^{\mu\nu} \leq \frac{3}{2}\eps^2_{\lambda}z+ 4 \eps_{\lambda}\ , 
\ee
which, using Cauchy-Schwartz, implies 
\begin{equation}\label{A_bound}
\int_{0}^{z}|A| \leq \sqrt{\frac{3}{2}\eps^2_{\lambda}z^2+ 4 \eps_{\lambda}z}\ . 
\end{equation}

Observe that \eqref{metric_prop} implies that, {taking any \textit{product} co-ordinate system on ${\cal M}_{\lambda}$ ({\it i.e.}, a co-ordinate of the form $(\alpha,\beta,z)$, where $z$ is as above and $\partial_\alpha,\partial_\beta$ are tangent to each surface orbit)}
\begin{equation}\label{metric_prop_cor1}
\partial_z g_{\alpha\alpha}=2A(\partial_{\alpha},\partial_{\alpha}) = 2A\left(\frac{\partial_{\alpha}}{||\partial_{\alpha}||},\frac{\partial_{\alpha}}{||\partial_{\alpha}||}\right)||\partial_{\alpha}||^2=2A\left(\frac{\partial_{\alpha}}{||\partial_{\alpha}||},\frac{\partial_{\alpha}}{||\partial_{\alpha}||}\right)\cdot g_{\alpha\alpha}\ .
\end{equation}
As $\frac{\partial_{\alpha}}{||\partial_{\alpha}||}$ is a unit vector 
\begin{equation}\label{A_bd_lie}
\left|A\left(\frac{\partial_{\alpha}}{||\partial_{\alpha}||},\frac{\partial_{\alpha}}{||\partial_{\alpha}||}\right)\right| \leq |A|\ ,
\end{equation}
so using this and \eqref{metric_prop_cor1},
\begin{equation}\label{metric_prop_cor}
|\partial_z g_{\alpha\alpha}| \leq 2|A|g_{\alpha\alpha}\ .
\end{equation}

Thus, for every {product} co-ordinate system on ${\cal M}_{\lambda}$, \eqref{metric_prop_cor} and \eqref{A_bound} imply that as long as 
\begin{equation}\label{small_d}
\eps_{\lambda}(z_2-z_1) \leq \frac{8}{3}
\end{equation}
 we have that
\be
\left| \log(g^{z_2}_{\alpha\alpha})-\log(g^{z_1}_{\alpha\alpha})\right|\leq 2\int_{z_1}^{z_2}|A| \leq 4\sqrt{2\eps_{\lambda}(z_2-z_1)}\ .
\ee
Exponentiating both sides, we obtain that for every tangent vector $W\in T_p\{z=z_2\}$ 
\begin{equation}\label{error_est}
e^{-4\sqrt{2\eps_{\lambda}(z_2-z_1)}} \leq \frac{g(W,W)}{g^{\mathrm{prod}}(W,W)} \leq e^{4\sqrt{2 \eps_{\lambda}(z_2-z_1)}}\ , 
\end{equation}
where $g^{\mathrm{prod}}$ is the product metric $dz^2+g^{z_1}$ under the standard flow lines. 
Thus, given some $\delta>0$, the two metrics remain a factor $(1+\delta)$ one from the other over a distance 
\begin{equation}\label{d_delta_def}
d^{\delta}_{\lambda}:=\frac{\log^2(1+\delta)}{32\eps_{\lambda}}={\frac{1}{K_\Lambda}}\frac{1}{64}\sqrt{\frac{3}{C_8}}\log^2(1+\delta)e^{\frac{1}{3}K_{\Lambda}^2\lambda}\ .
\end{equation}
Note that for $\delta$ sufficiently small, this is compatible with the assumption \eqref{small_d} which was previously employed.

\begin{claim}
For every $\delta$ there exists some $1\gg\rho=\rho(\delta)>0$ such that for every $\lambda>0$,  each strip ${\cal M}^{z_1,\delta}_{\lambda}:=\{x\in {\cal M}_{\lambda}\;|\;  z_1 \leq z(x) \leq z_1+d^{\delta}_{\lambda}\}$ satisfies
\begin{equation}
\frac{\mathrm{Vol}({\cal M}_{\lambda}^{z_1,\delta})}{V(\lambda)} \geq \rho\ . 
\end{equation} 
\end{claim}
\begin{proof}
Using \eqref{area_est1} and \eqref{vol_est1}, we see that 
\be
\mathrm{Area}(\{z=s\})\geq \frac{1}{2}S_{\mathrm{min}}(0)e^{\frac{2}{3}K_{\Lambda}^2\lambda}\ ,
\ee
and
 \be
 V(\lambda) \leq 2V(0)e^{K_{\Lambda}^2\lambda}\ .
 \ee
Since,
\begin{equation}\label{co_area}
\mathrm{Vol}({{\cal M}_{\lambda}})= \int_0^{z_0} {ds}\;\mathrm{Area}(\{z=s\})
\end{equation}
the result holds with 
\begin{equation}\label{rho_def}
\rho(\delta):={\frac{1}{64} \sqrt{\frac{3}{C_8}}}\frac{\log^2(1+\delta)S_{\mathrm{min}}(0)}{4V(0) { K_\Lambda}}\ .
\end{equation}

\end{proof}

\subsection{\label{sec:L2cond}Conditional $L^2$ closeness to exponentially expanding slices}  

{For technical reasons, it will be important in the following} to define norms with respect to the comparison metric ${\mathbf g}$ defined in \eqref{comp_metric}, instead of the actual metric $g$. In this Section, we are going to deduce results under the condition the two metrics are a priori close to each other. We are going to relax this assumption in the following Section. To compare norms defined with respect to the two different metrics we will need the following lemma.

\begin{lemma}\label{lemma:metricnormbound} There exists a $1>\gamma_0>0$ with the following property. Suppose $0<\gamma\leq \gamma_0$ and 
\be\label{eq:neammetric}
||g-\mathbf{g}||_{\mathbf{g}} \leq \gamma\ . 
\ee 
Then, for any $2$-tensor $T$, there exists a universal constant $D_3$ such that 
\begin{equation}\label{Twithtwogs}
\left(1-D_3\gamma\right)||T||_{\mathbf{g}}  \leq ||T||_{g} \leq \left(1+D_3\gamma\right)||T||_{\mathbf{g}}\ .
\end{equation}
\begin{proof}
Choose coordinates at a point such that $\mathbf{g}_{ij}=\delta_{ij}$, {\it i.e.}~orthogonal at the point. Then the condition (\ref{eq:neammetric}) implies that $g_{ij}=\delta_{ij}+\eps_{ij}$ where $|\eps_{ij}|\leq \gamma$. By the inversion formula for $3\times 3$ matrices, $g^{ij}=\delta^{ij}+\tilde{\eps}^{ij}$ where $|\tilde{\eps}^{ij}|\leq D_1\gamma$ (for sufficiently small $\gamma$).
Thus, 
\begin{equation}\label{converstion_1}
||T||^2_{g}=g^{ij}g^{kl}T_{ik}T_{jl}=\delta^{ij}\delta^{kl}T_{ik}T_{jl}+\eps_{ijkl}T_{ik}T_{jl}\ ,
\end{equation}
where $|\eps_{ijkl}|\leq D_2\gamma$. Now, the first term in the right hand side is (by definition, and by our choice of co-ordinates) $||T||^2_{\mathbf{g}}$.
Moreover, again by definition 
\be
|T_{ik}|^2 \leq ||T||^2_{\mathbf{g}}\ , 
\ee
as this is one of the terms appearing in the sum (again, by our choice of co-ordinates). Thus 
\be
|\eps_{ijkl}T_{ik}T_{jl}| \leq D_3\gamma ||T||^2_{\mathbf{g}}\ ,
\ee 
from which, using \eqref{converstion_1},  we obtain \eqref{Twithtwogs}.
\end{proof}
\end{lemma}

\begin{lemma}
 Let $\gamma_0$ be as in Lemma~\ref{lemma:metricnormbound}, and let $0<\gamma\leq\gamma_0$. There exists a positive constant $C_9$ with the following significance: let  $\lambda_0>0$ be some time and $\mathbf{g}$ defined in eq.~\eqref{comp_metric}
and set, for each $\lambda>\lambda_0$
\begin{equation}\label{real_E_def}
E(\lambda):=\int_{{\cal M}_{\lambda}}||g(\lambda)-\mathbf{g}(\lambda)||^2_{\mathbf{g}(\lambda)}dV_{g(\lambda)}\ .
\end{equation}
Then, for every $\lambda$ such that $||g(\lambda)-\mathbf{g}(\lambda)||_{\mathbf{g}(\lambda)} \leq \gamma$, we have that 
\begin{equation}\label{DEfinal}
 E'(\lambda) \leq K_\Lambda^2 \left(1+ \frac32 C_2 e^{-\frac{2}{3}K_\Lambda^2 \lambda}\right)E(\lambda)+ C_9 K_\Lambda^{1/2}e^{\frac{1}{6}K_{\Lambda}^2\lambda}E(\lambda)^{1/2} \;.
\end{equation} 
\end{lemma}

\begin{proof}
We have (we suppress the dependence on $\lambda$ in $g$ and $\mathbf{g}$),
\begin{align}
 E'(\lambda)=&\int_{\cal M_\lambda} K^2||g-\mathbf{g}||^2_{\mathbf{g}} dV_g+2\int_{\cal{M}_\lambda}\langle g-\mathbf{g}, 2KK_{ij}-\frac{2}{3}K_{\Lambda}^2\mathbf{g} \rangle_{\mathbf{g}} dV_g  -\frac{4}{3} K_\Lambda^2 E(\lambda)\ .
\end{align}
To get to the final inequality, we bound the first two terms on the RHS separately~\footnote{Notice that if we had defined the norms with respect to the metric $g$ instead of $\mathbf{g}$, this equation would contain terms involving the $\lambda$-derivative of $g$ that would be difficult to control. This is the reason of choosing norms with respect to $\mathbf{g}$.}. For the first one, we use Theorem~\ref{th:boundonmax} ($C_2$ is defined in \eqref{pointwise_scalar_low}): 
\begin{equation}\label{DEfirst}
 \int_{\cal{M}_\lambda} dV_g\; K^2||g-\mathbf{g}||^2_\mathbf{g}\leq  K_{\Lambda}^2 \left(1+ \frac32 C_2 e^{-\frac{2}{3}K_{\Lambda}^2\lambda}\right)E(\lambda)\ .
\end{equation}
For the second one, writing $K_{ij}=\frac{K}{3}g_{ij}+\sigma_{ij}$, we first write
\be\label{eq:longterm}
2\int_{\cal{M}_\lambda} \langle g-\mathbf{g}, 2KK_{ij}-\frac{2}{3}K_{\Lambda}^2\mathbf{g} \rangle_\mathbf{g} dV_g \leq  \frac{4}{3}\int_{\cal{M}_\lambda} \langle g-\mathbf{g}, K^2g-K_{\Lambda}^2\mathbf{g} \rangle_\mathbf{g} dV_g+4\left|\int_{\cal{M}_\lambda} \langle g-\mathbf{g}, K\sigma_{ij} \rangle_\mathbf{g} dV_g\right|\ .
\ee
The first term on the right-hand side of (\ref{eq:longterm}) can be bounded as  
\begin{align}\label{DEsecond}
& \frac{4}{3}\int_{\cal{M}_\lambda} \langle g-\mathbf{g}, K^2g-K_{\Lambda}^2\mathbf{g} \rangle_\mathbf{g} dV_g\leq \frac{4}{3} K_{\Lambda}^2E(\lambda) + \frac{4}{3} \int_{\cal{M}_\lambda} |\langle g-\mathbf{g}, (K^2-K_{\Lambda}^2)g\rangle_\mathbf{g}| dV_g \\ \nonumber
&\leq \frac{4}{3} K_{\Lambda}^2E(\lambda) + \frac{8}{3} K_{\Lambda}\left(1+\frac{C_1}{2}e^{-\frac{2}{3}K_{\Lambda}^2 \lambda}\right)\int_{\cal{M}_\lambda} |K-K_{\Lambda}|\cdot|\langle g-\mathbf{g}, g\rangle_\mathbf{g}| dV_g \\\nonumber
&\leq \frac{4}{3} K_{\Lambda}^2E(\lambda) + \frac{8}{3}K_{\Lambda}\left(1+\frac{C_1}{2}e^{-\frac{2}{3}K_{\Lambda}^2 \lambda}\right)\sqrt{3(1+D_3 \gamma)}E(\lambda)^{1/2}\left(\int_{\cal{M}_\lambda}|K-K_{\Lambda}|^2\right)^{1/2} \\\nonumber
&\leq \frac{4}{3} K_{\Lambda}^2E(\lambda) + \frac{8}{3}K^{1/2}_{\Lambda}\left(1+\frac{C_1}{2}\right)\sqrt{3(1+D_3 \gamma)}C_7^{1/2}E(\lambda)^{1/2}e^{\frac{1}{6}K_{\Lambda}^2\lambda }\ ,
\end{align}
where in the second step we used the bound \eqref{eq:maxbound}, in the third the Cauchy-Schwartz inequality (both on the integral and on the scalar product with respect to the metric $\mathbf{g}$) and Lemma~\ref{lemma:metricnormbound}; in the last step we used the inequality $(K-K_\Lambda)^2 \le |K^2-K_\Lambda^2|$, for $K\ge 0$, and the bound~\eqref{quant_est}.
The second term on the right-hand side of \eqref{eq:longterm} is bounded by  
\begin{align}\label{DEthird}
 4\left|\int_{\cal{M}_\lambda} \langle g-\mathbf{g}, K\sigma_{ij} \rangle_\mathbf{g} dV_g\right|   &\leq 4 K_{\Lambda}(1+C_1e^{-\frac{2}{3}K_{\Lambda}^2\lambda})E(\lambda)^{1/2}||\sigma_{ij}||_{L^2_\mathbf{g}}\\ \nonumber &\leq 4 C_7^{1/2} (1 + D_3 \gamma) K_{\Lambda}^{1/2} (1+C_1 )e^{\frac{1}{6}K_{\Lambda}^2\lambda }E(\lambda)^{1/2},
\end{align}
where we used the bound \eqref{eq:maxbound}, the Cauchy-Schwartz inequality, the bound \eqref{quant_est} and the inequality $||\sigma_{ij}||_{L^2_\mathbf{g}}\leq \left(1+D_3\gamma\right)||\sigma_{ij}||_{L^2_{g}}$ by Lemma~\ref{lemma:metricnormbound}. 

Assuming that $\lambda\geq \lambda_0$ and $||g(\lambda)-\mathbf{g}(\lambda)||_{\mathbf{g}(\lambda)}\leq \gamma$ one can put together \eqref{DEfirst}, \eqref{DEsecond} and \eqref{DEthird} to get the final inequality \eqref{DEfinal} with a suitable constant $C_9$ that can be expressed in terms of the constant that appear \eqref{DEsecond} and \eqref{DEthird}.
\end{proof}

\begin{lemma}\label{E_grow_lemma}
Let $\gamma_0$ be as in Lemma~\ref{lemma:metricnormbound}, and let $0<\gamma\leq\gamma_0$. There exists a {universal} constant $C_{10}<\infty$ with the following significance: let  $\lambda_0>0$ be some time. Define $\mathbf{g}$ and $E$ as in equations \eqref{comp_metric}, \eqref{real_E_def}.  Let $\lambda>\lambda_0$ be such  that for every $\lambda' \in [\lambda_0,\lambda]$ we have that $||g(\lambda')-\mathbf{g}(\lambda')||_{\mathbf{g}(\lambda')}<\gamma$. Then
\begin{equation}\label{E_growth}
E\leq { \frac{C_{10}}{K_\Lambda^3}} e^{-\frac{2}{3}K_{\Lambda}^2\lambda_0}e^{K_{\Lambda}^2\lambda}\ .
\end{equation}
\end{lemma}
\begin{proof}
Making the substitution $\tilde{E}(\lambda)=e^{\frac{2}{3}K_{\Lambda}^2\lambda_0}e^{-K_{\Lambda}^2\lambda}E(\lambda)$, the inequality \eqref{DEfinal} becomes
\begin{align}\label{E_upper}
&\tilde E'(\lambda) \leq \frac32 C_2 K_\Lambda^2 e^{-\frac{2}{3}K_\Lambda^2 \lambda} \tilde E(\lambda)+ C_9  K_\Lambda^{1/2} e^{-\frac{1}{3}K_{\Lambda}^2\lambda}e^{\frac{1}{3}K_{\Lambda}^2\lambda_0}\tilde E(\lambda)^{1/2} \\ \nonumber
&\qquad\leq { C_{10}' K_\Lambda^2 e^{-\frac{1}{3}K_{\Lambda}^2\lambda}e^{\frac{1}{3}K_{\Lambda}^2\lambda_0}\left(\tilde E(\lambda)+\frac{C_9}{2 C_{10}'K_\Lambda^3}\right)} \ ,
\end{align}
{where for simplicity in the second step we assumed $C_9 >0$ (\footnote{If $C_9 =0$, the first inequality of eq.~\eqref{E_upper} implies $\tilde E(\lambda) =0$ so that the Lemma holds with $C_{10}=0$.}) and we used that $\sqrt{a b} \le (a+b)/2$, with $a>0, b>0$, and where $C_{10}' :=\frac32 C_2 +\frac{1}{2}C_9$.}

Thus, 
\begin{equation}
\log\left(\tilde{E}+\frac{C_9 }{2C_{10}'K_\Lambda^3}\right)' \leq  K_\Lambda^2 C'_{10} e^{-\frac{1}{3}K_{\Lambda}^2\lambda}e^{\frac{1}{3}K_{\Lambda}^2\lambda_0}\ ,
\end{equation}
which, together with $\tilde{E}(\lambda_0)=0$,  integrates to {
\begin{equation}
\log\left(\tilde{E}+\frac{C_9 }{2 C_{10}'K_\Lambda^3}\right) \leq 3C'_{10} -\log\left(2 C_{10}' K_\Lambda^3/C_9\right)
\end{equation}}
for all $\lambda$. Thus
\begin{equation}
\tilde{E}(\lambda) \leq {\frac{C_{10}}{K_\Lambda^3}}:={ \frac{{C_9}e^{3C'_{10}}}{2 C_{10}' K_\Lambda^3}}
\end{equation}
which is equivalent to \eqref{E_growth}.

\end{proof}

\subsection{Unconditional pointwise closeness to {exponentially expanding} slices}\label{uncond_point} 
In this Section we put together the results on the $z$-dependence of the spatial metric obtained in Section \ref{sec:zdep} with the results on the $L^2$-closeness of Section \ref{sec:L2cond} in order to prove the pointwise convergence of the metric to the exponentially expanding comparison metric \eqref{comp_metric}. 

\begin{theorem}\label{main_spatial_thm}
There exists some $\lambda_{\ast}<\infty$ and a universal constant $C_{11}<\infty$ such that for every $\lambda_0>\lambda_{\ast}$, defining, as before,  $\mathbf{g}(\lambda)$ as in  \eqref{comp_metric},
we have
\begin{equation}\label{finalprox}
||g(\lambda)-\mathbf{g}(\lambda)||_{ \mathbf{g}} \leq C_{11}e^{-\frac{1}{6}K_{\Lambda}^2\lambda_0}\ ,
\end{equation}
pointwise for every $\lambda>\lambda_0$.
\end{theorem}

\begin{proof} 
Let $\delta_\ast=\gamma/4$, and define $\lambda_\ast$ so that
\begin{equation}\label{lambdastar_choice}
2 { \frac{C_{10}}{K_\Lambda^3}} e^{-\frac23 K_{\Lambda}^2\lambda_\ast} =\frac{\delta_\ast^2\rho(\delta_\ast) }{4} V(0)\ ,
\end{equation}
where $C_{10}$ is the constant that appears in eq.~\eqref{E_growth} and $\rho(\delta)$ is in \eqref{rho_def}.
Let $\lambda_0>\lambda_\ast$, and define $\delta$ to be the solution of
\begin{equation}\label{lambda0_choice}
2 { \frac{C_{10}}{K_\Lambda^3}} e^{-\frac23 K_{\Lambda}^2\lambda_0} =\frac{\delta^2\rho(\delta) }{4} V(0)\ .
\end{equation}
Notice that since $\delta$ is a monotonically decreasing function of $\lambda_0$, we automatically have that for $\lambda_0>\lambda_\ast$, $\delta<\gamma/4$.

Eq.~\eqref{E_growth} and \eqref{vol_est1} imply that 

\begin{equation}\label{E_small_vol}
E(\lambda) \leq \frac{\delta^2\rho(\delta) }{4} V(\lambda)\ ,
\end{equation}
as long as 

\begin{equation}\label{real_small_ass}
||g(\lambda')-\mathbf{g}(\lambda')||_{\mathbf{g}(\lambda')}\leq \gamma \;\textrm{ for every\;} \lambda'\in [\lambda_0,\lambda]\ .
\end{equation}
Recall also that for $\lambda'=\lambda_0$, $||g(\lambda')-\mathbf{g}(\lambda')||_{\mathbf{g}(\lambda')}=0$, and note moreover that this norm is a continuous function of $\lambda'$.

Now, suppose for the sake of contradiction that there exists some $\lambda''>\lambda_0$ such that
\be
\max\left(||g(\lambda'')-\mathbf{g}(\lambda'')||_{\mathbf{g}(\lambda'')}\right)\geq  4\delta. 
\ee  
 Let $\lambda$ be the infimum of the $\lambda''$, and notice that $\lambda>\lambda_0$. Let $z_{\rm bad}$ be a point where this maximum is obtained at $\lambda$, so that at $z_{\rm bad}$ the following holds
\be
||g(\lambda)-\mathbf{g}(\lambda)||_{\mathbf{g}(\lambda)}= 4\delta\ . 
\ee  
In particular, we have that $||g(\lambda)-\mathbf{g}(\lambda)||_{\mathbf{g}(\lambda)}\leq 4\delta < \gamma$ for every $z$. Note that \eqref{real_small_ass} is satisfied up to time $\lambda$, so in particular, \eqref{E_small_vol} is valid at time $\lambda$. 
Now, applying $\eqref{error_est}$ twice, at flow times $\lambda_0$ and $\lambda$, starting at $z_{\mathrm{bad}}$, we see that $||g(\lambda)-\mathbf{g}(\lambda)||_{\mathbf{g}(\lambda)}\geq \delta$ for every $z$ in ${\cal M}_{\lambda}^{z_{\mathrm{bad}}, \delta}$. Thus
 
 \begin{equation}
 E(\lambda) \geq \int \displaylimits_{{\cal M}_{\lambda}^{z_{\mathrm{bad}},\delta}} ||g(\lambda)-\mathbf{g}(\lambda)||^2_{\mathbf{g}(\lambda)} \geq  \delta^2 \mathrm{Vol}({\cal M}_{\lambda}^{z_{\mathrm{bad}},\delta})\geq \delta^2\rho( \delta)V(\lambda)\ ,
 \end{equation}
 which contradicts \eqref{E_small_vol}.  Therefore, there cannot exist such a $\lambda$, and so $||g(\lambda)-\mathbf{g}(\lambda)||_{\mathbf{g}(\lambda)}<4\delta$ always. The dependence of $\delta$ on $\lambda_0$ can be read from \eqref{lambda0_choice} and \eqref{rho_def} by Taylor expansion 
\begin{equation}\label{depend_delta_lambda}
\delta \sim e^{-\frac16 K_{\Lambda}^2\lambda_0} \;.
\end{equation}
This gives the final result \eqref{finalprox}.
\end{proof}
{In the following we are going to often assume that the distance of eq.~\eqref{finalprox} is small, say $< \frac{1}{100}$, by imposing $\lambda_0 > 6 K_\Lambda^{-2} \log(100 C_{11})$.}

{
\subsection{Closeness to de Sitter slices over exponentially expanding balls}\label{prox_glob_flat_sec}
{In this section we want to prove the pointwise convergence over expanding balls of the spatial metric to the spatial metric of de Sitter space in flat slicing:
\begin{equation}\label{slice_ds_def}
\mathbf{g}_{\mathrm{dS}}(\lambda) :=e^{\frac{2}{3}K_{\Lambda}^2(\lambda-\lambda_0)} g_{\mathrm{Euc}}\ ,
\end{equation}
with $g_{\mathrm{Euc}}$ the flat Eucliden 3d metric. {The idea is to prove that the spatial metric $g(\lambda_0)$ for large $\lambda_0$ becomes approximately flat since the surface orbits have larger and larger area and at the same time the metric becomes independent of the orthogonal direction $z$. The exponential growth factor in $\lambda$ is then fixed using the results in the previous section.}
}

Let $\rho_1$ be the shortest length of a non-contractible loop in one of the surface orbits contained in $\mathcal{M}_{\lambda_{\ast}}$ (set $\rho_1=1{ /K_\Lambda}$ if $\Sigma$ is a sphere). By Theorem \ref{main_spatial_thm} {at $\lambda_*$ the metric $g$ and $\mathbf{g}$ differ by $< \gamma/4 < 1/4$},}  so that every curve in a surface orbit in $\mathcal{M}_{\lambda_0}$, $\lambda_0 > \lambda_*$, of length ${<} \frac{1}{2}\rho_1e^{\frac{1}{3}K_{\Lambda}^2(\lambda_0-\lambda_{\ast})}$, is contractible. Recalling that $\chi$ is the Euler characteristic of the orbit surfaces, from~(\ref{area_est1}) we further know that each orbit surface has sectional curvature 
\begin{equation}\label{eq:sec1}
|\mathrm{Sec}| \leq C_{\mathrm{Sec}}^2{K_{\Lambda}^2} e^{-\frac{2}{3}K_{\Lambda}^2\lambda_0}\ ,
\end{equation}  
where
\begin{equation}
C_{\mathrm{Sec}}:=\sqrt{\frac{4\pi |\chi|}{S_{\min}(0)K_{\Lambda}^2}}\ .
\end{equation}

{ If we consider the standard forms of the metric in polar coordinates for 2-sphere, 2-plane and 2-hyperboloid (for the sphere, for instance, this is $dr^2+{\frac{\sin^2(\mathrm{Sec}^{1/2}\cdot r)}{\mathrm{Sec}}}d\theta^2$), it is useful to notice that, for $K_{\Lambda} r \leq 2e^{\frac{1}{12}K_{\Lambda}^2\lambda_0}$ and choosing $\lambda_0\geq\frac{4}{K_\Lambda^2}\log(2 C_{\mathrm{Sec}})$, we can use that, for $0\leq t\leq 1$, $\left|\frac{\sin^2(t)}{t^2}-1\right|\leq  t^2$, and $\left|\frac{\sinh^2(t)}{t^2}-1\right|\leq t^2$, to write:} 
\begin{equation}\label{coeff_com}
\left|\frac{\sin^2\left(C_{\mathrm{Sec}}e^{-\frac{1}{3}K_{\Lambda}^2\lambda_0}K_{\Lambda}r\right)}{C_{\mathrm{Sec}}^2e^{-\frac{2}{3}K_{\Lambda}^2\lambda_0}K_{\Lambda}^2r^2}-1\right| {\leq}  {4} C_{\mathrm{Sec}}^2 e^{-\frac{1}{{{2}}}K_{\Lambda}^2\lambda_0}\ .  
\end{equation} 
The same bound holds if we replace $\sin$ by $\sinh$. {It is now useful to impose $\lambda_0 \geq \lambda'_{\ast\ast}$, where $\lambda'_{\ast\ast}$ is given by}
\begin{equation}\label{lambda_ast_ast_def}
\lambda_{\ast\ast}'=\max\left({\lambda_*}, { \frac{4}{K_\Lambda^2}\log(2 C_{\mathrm{Sec}})},\frac{4\log\left(\frac{{ 20}}{{ 3}\rho_1K_{\Lambda}}\right)+\frac{4}{3}K_{\Lambda}^2\lambda_{\ast}}{K_{\Lambda}^2}\right)\ .
\end{equation}
{The first term on the r.h.s. was imposed above (\ref{eq:sec1}) to set a maximum length for the contractible curves; the second term on the r.h.s. was imposed above~(\ref{coeff_com}) to ensure that the argument of the Sine on the l.h.s of~(\ref{coeff_com}) is at most equal to one; the third condition ensures that
for every $\lambda_0 \geq \lambda'_{\ast\ast}$, and every point $p\in \mathcal{M}_{\lambda_0}$, if $p\in \Sigma$ for some orbit surface $\Sigma$, then $\exp_p^{\Sigma}$ maps the {2-dimensional} Euclidean ball $B_{ 2}(0,\frac{5}{3K_{\Lambda}}e^{\frac{1}{12}K_{\Lambda}^2\lambda_0})$ diffeomorphically to the intrinsic ball in $\Sigma$,  $B^{\Sigma}(p,\frac{5}{3K_{\Lambda}}e^{\frac{1}{12}K_{\Lambda}^2\lambda_0})$ (in fact, the first term on the r.h.s.~of~(\ref{lambda_ast_ast_def}) ensures that the diameter of this ball is shorter that the shortest non-contractible curve in $\Sigma$).} Moreover, setting 
\begin{equation}
g^{\Sigma}_{\mathrm{Euc}}=\left(\exp^{\Sigma}_{p}\right)_\ast(\langle\cdot , \cdot\rangle)\ ,  
\end{equation}
\eqref{coeff_com} implies that for every tangent vector $W$ to $\Sigma$,
\begin{equation}\label{comp_euc}
 1-{4} C_{\mathrm{Sec}}^2 e^{-\frac{1}{{{2}}}K_{\Lambda}^2\lambda_0}\leq \frac{g(W,W)}{g^{\Sigma}_{\mathrm{Euc}}(W,W)} \leq  1+{4} C_{\mathrm{Sec}}^2 e^{-\frac{1}{{{2}}}K_{\Lambda}^2\lambda_0}\ .
\end{equation}
{Let us set 
\be\label{deltachoice}
\delta={\frac{1}{10}}e^{-\frac{1}{12}K_{\Lambda}^2\lambda_0} \ .
\ee
 Then for $\lambda_0> \frac{12}{5{K_\Lambda^2}}\log\left({ 40} C_{\mathrm{Sec}}^2\right)$, we can ensure that ${4} C_{\mathrm{Sec}}^2 e^{-\frac{1}{{{2}}}K_{\Lambda}^2\lambda_0}<\delta$, so that:
\begin{equation}\label{comp_euc}
 1-\delta\leq \frac{g(W,W)}{g^{\Sigma}_{\mathrm{Euc}}(W,W)} \leq  1+\delta\ .
\end{equation}
} Now, let $g$ be the true metric on $\mathcal{M}_{\Sigma}$ and   $g_{\mathrm{Euc}}$ be the Euclidean product metric $g_{\mathrm{Euc}}:=dz^2+g^{\Sigma}_{\mathrm{Euc}}$. Then \eqref{comp_euc}, \eqref{error_est}, \eqref{d_delta_def} and \eqref{lambda_ast_ast_def} imply that $g$ and $g_{\mathrm{Euc}}$ are a factor of $(1+\delta)$ from one another over an interval (in the $z$-direction) of length $d_{\lambda_0}^{\frac{\delta}{2}}$. { Therefore, taking
\begin{equation}\label{eq:lambda''}
\lambda''_{\ast\ast}=\max\left(\lambda'_{\ast\ast},{\frac{12}{5{K_\Lambda^2}}\log\left({40}C_{\mathrm{Sec}}^2\right),} {\frac{12}{5 K_\Lambda^2}\log\left({64 \cdot \frac{5}{3}}\cdot 36\cdot 10^2 \sqrt{\frac{C_8}{3}}\right)}\right)\ ,
\end{equation}
we get that for every $\lambda_0 {>} \lambda_{\ast\ast}''$, for every tangent vector $W\in T_{q}{\mathcal{M}}_{\lambda_0}$  at a point $q\in{\cal M}_{\lambda_0}\cap B^{(\mathcal{M}_{\lambda},g(\lambda))}(p,\frac{5}{3K_{\Lambda}}e^{\frac{1}{12}K_{\Lambda}^2\lambda_0})$, we have
\begin{equation}\label{comp_euc_3}
1 -\delta \leq \frac{g|_q(W,W)}{g_{\mathrm{Euc}}|_q(W,W)} \leq { 1+\delta} \ .
\end{equation}
{The second factor on the r.h.s. in (\ref{eq:lambda''}) was imposed just above (\ref{comp_euc});  the last factor in (\ref{eq:lambda''}) comes from imposing that the distance $d_{\lambda_0}^{\frac{\delta}{2}}$ in \eqref{d_delta_def} is larger than the radius of the ball above: $\frac{5}{3K_{\Lambda}}e^{\frac{1}{12}K_{\Lambda}^2\lambda_0}$.}

We can now prove the convergence to the de Sitter metric \eqref{slice_ds_def}. We define
\begin{equation}
\lambda_{\ast\ast}=\max\left(\lambda''_{\ast\ast},\frac{12\log(10 C_{11})}{K_{\Lambda}^2}\right)\ ,
\end{equation}}
(the second condition guarantees that the error of Theorem~\ref{main_spatial_thm}, $C_{11} \exp(-1/6 \cdot K_\Lambda^2 \lambda_0)$ is smaller than the $\delta$ defined in eq.~\eqref{deltachoice}) we have
\begin{theorem}\label{main_spatial_thm_2}
For every $\lambda_0>\lambda_{\ast\ast}$, we have
\begin{equation}\label{finalprox_2}
||g(\lambda)-\mathbf{g}_{\mathrm{dS}}(\lambda)||_{g(\lambda)} {<} { 16} e^{-\frac{1}{12}K_{\Lambda}^2\lambda_0}\ ,
\end{equation}
pointwise for every $\lambda>\lambda_0$ on $B^{(\mathcal{M}_{\lambda},g(\lambda))}\left(p_{\lambda}, \frac{1}{K_{\Lambda}}e^{\frac{1}{12}K_{\Lambda}^2\lambda_0}\cdot e^{\frac{1}{3}K_{\Lambda}^2(\lambda-\lambda_0)}\right)$.  Here $p_{\lambda}$ results from following $p$ along the flow.
\end{theorem}
{
\begin{proof} 
Remember that at $\lambda_0$, $\mathbf{g}(\lambda_0)=g(\lambda_0)$. Therefore,  (\ref{comp_euc_3}) gives (suboptimally as usual)
\be\label{eq:modeldslambda0}
||{\mathbf{g}(\lambda_0)}-\mathbf{g}_{\mathrm{dS}}(\lambda_0)||_{\mathbf{g}(\lambda_0)}{ <} 4 \delta\
\ee
on $B^{(\mathcal{M}_{\lambda_0},\mathbf{g}(\lambda_0))}\left(p_{\lambda}, \frac{1}{K_{\Lambda}}e^{\frac{1}{12}K_{\Lambda}^2\lambda_0}\right)$. (Notice that we took a ball of radius smaller than above.) Since both metrics evolve with time with the same rescaling, $\mathbf{g}(\lambda)=e^{\frac{2}{3}K_\Lambda^2(\lambda-\lambda_0)}\mathbf{g}(\lambda_0)$ and $\mathbf{g}_{\mathrm{dS}}(\lambda)=e^{\frac{2}{3}K_\Lambda^2(\lambda-\lambda_0)}\mathbf{g}_{\mathrm{dS}}(\lambda_0)$, (\ref{eq:modeldslambda0}) is true at all times:
\be\label{eq:modeldslambda}
||{\mathbf{g}(\lambda)}-\mathbf{g}_{\mathrm{dS}}(\lambda)||_{\mathbf{g}(\lambda)}< 4\delta 
\ee
on $B^{(\mathcal{M}_{\lambda_0},\mathbf{g}(\lambda))}\left(p_{\lambda}, \frac{1}{K_{\Lambda}}e^{\frac{1}{12}K_{\Lambda}^2\lambda_0}\cdot e^{\frac{1}{3}K_{\Lambda}^2(\lambda-\lambda_0)}\right)$. But, from Theorem~\ref{main_spatial_thm}, we have that, at all times $\lambda\geq\lambda_0$
\be
||g(\lambda)-{\mathbf{g}}(\lambda)||_{{g}(\lambda)}<2||g(\lambda)-{\mathbf{g}}(\lambda)||_{{\mathbf{g}}(\lambda)}<8 \delta\ .
\ee
Therefore, we have, on $B^{(\mathcal{M}_{\lambda},g(\lambda))}\left(p_{\lambda}, \frac{1}{K_{\Lambda}}e^{\frac{1}{12}K_{\Lambda}^2\lambda_0}\cdot e^{\frac{1}{3}K_{\Lambda}^2(\lambda-\lambda_0)}\right)$ that
\be
||g(\lambda)-{\mathbf{g}}_{\mathrm{dS}}(\lambda)||_{{g}(\lambda)}<2||g(\lambda)-{\mathbf{g}}_{\mathrm{dS}}(\lambda)||_{{\mathbf{g}}(\lambda)}<4||{\mathbf{g}}(\lambda)-{\mathbf{g}}_{\mathrm{dS}}(\lambda)||_{{\mathbf{g}}(\lambda)}<16 \delta\ ,
\ee
as we wished to show.
\end{proof}}}

\section{Space-time closeness}\label{space_time_sec}

We are now ready to show that, asymptotically, the spacetime becomes close to de Sitter space, in the sense that {the length of any future-oriented, timelike or null curve between two spacetime points} approaches the one evaluated  between the same points using the de Sitter metric, once both points are taken at late enough times.

To achieve our purpose, we need to gain some additional control on the extrinsic curvature, which we do first~\footnote{ By making stronger assumptions on the geometry of $M^{(3+1)}$, it is possible to obtain a stronger conclusion on this aspect, which however does not alter the physical equivalence to de Sitter space that we discuss in the last section. It will be discussed in a future publication~\cite{inprogressshort}.}. Let us start by {noticing that \eqref{eq:MCF} and  SEC imply:}
 \be\label{K_ev}
\frac{d K}{d \lambda} -\Delta K + \frac{1}{3} K\left(K^2-K_\Lambda^2\right) + \sigma^2 K \leq 0 \;.
\ee
 Let let us now observe that we can bound $ \int_{{\cal M}_{\lambda}} dV\; |\nabla K|^2$ if this is integrated over a finite flow-time interval. Specifically, we have
\begin{lemma}\label{lemma:gradK}
For every $\lambda>0$,
\begin{equation}
\int_{\lambda}^{\lambda+\frac{1}{K_\Lambda^2}}d\lambda' \int_{{\cal M}_{\lambda'}} dV\; |\nabla K|^2 \leq \frac{C'_{12}}{K_\Lambda} \;e^{\frac{1}{3}K_{\Lambda}^2\lambda}.
\end{equation}
\end{lemma}
\begin{proof}
Using, first  \eqref{K_ev},  and then  \eqref{eq:maxbound} {and \eqref{quant_est}}, we compute 
\begin{align}\label{eq:gradKproof}
\frac{d}{d\lambda}&\int_{{\cal M}_{\lambda}}dV\;\left(K^2-K_{\Lambda}^2\right)=\int_{{\cal M}_{\lambda}}dV\; \left((K^2-K_{\Lambda}^2)K^2+ 2K\frac{dK}{d\lambda}\right)  \\ \nonumber
&\leq \int_{{\cal M}_{\lambda}}dV\left[(K^2-K_{\Lambda}^2)K^2+ 2K\left(\Delta K+\frac{1}{3}K(K_{\Lambda}^2-K^2)-\sigma^2K \right)\right]\\ \nonumber
& \leq  \frac{2}{3} C_7\, K_\Lambda\;e^{\frac{1}{3}K_{\Lambda}^2\lambda}-2\int_{{\cal M}_{\lambda}} |\nabla K|^2dV\ ,
\end{align}
where we used that $K^2\leq 2 K_\Lambda^2$ for the flow times that we are considering {(which follows from (\ref{eq:C4bound}) or (\ref{eq:closedbound}) and our redefinition of the zero flow time).}
Integrating this over $[\lambda,\lambda+1/K_\Lambda^2]$ gives the desired result with {$C'_{12}=\left( \left(e^{1/3}-1\right)+ \frac 12\left(e^{1/3}+1\right) \right)C_7$}.
\end{proof}

Observe that Lemma \ref{lemma:gradK} and \eqref{area_est1} imply that 

\begin{equation}\label{z_grad}
\int_{\lambda}^{\lambda+\frac{1}{K_\Lambda^2}}d\lambda' \int_0^{L(\lambda)} dz\; |\nabla K|^2 \leq C_{12}K_{\Lambda} \;e^{-\frac{1}{3}K_{\Lambda}^2\lambda}\ ,
\end{equation}
where $C_{12}=\frac{2}{S_{\min}(0)K^2_\Lambda}C'_{12}$.

Similarly, \eqref{quant_est} and \eqref{area_est1} give

\begin{equation}\label{z_quant_est}
\int_{0}^{L(\lambda)}dz \; |K^2-K_{\Lambda}|^2 \leq C_{13}K_{\Lambda}e^{-\frac{1}{3}K_{\Lambda}^2\lambda}\ ,
\end{equation}
where $C_{13}=\frac{2}{S_{\min}(0)K_{\Lambda}^2}C_7$.

We are now going to show that this result allows us to say that $K$ is pointwise close to $K_\Lambda$ at most of the late-enough flow  times. In fact, \eqref{z_grad} guarantees that, at most flow-times, {$\int_{0}^{L(\lambda)} dz\; |\nabla K|^2 $} is small, but there can still be a small set of flow times where this quantity is badly behaved. For each integer $i\geq 0$ let us therefore identify the set of flow times, $B_i$, within the interval {$[\frac{i}{K_\Lambda^2},\frac{i+1}{K_\Lambda^2})$}  when there is no good gradient bound:
\begin{equation}
B_i:=\{K_\Lambda^2\lambda \in [i,i+1)\;|\;  \int dz\; |\nabla K(\lambda,z)|^2 \geq { C_{12}\,K_\Lambda^{3}} e^{-{ \frac{2}{9}}i}\}\ .
\end{equation}
Because of \eqref{z_grad}, $B_i$ has measure (length) satisfying the estimate
\begin{equation}
|B_i| \leq { \frac{1}{K_\Lambda^2}} e^{-\frac{{ 1}}{9}i}\ . 
\end{equation}
Given an integer $i_0\geq 0$, denote by ${\mathcal{B}_{i_0}=\bigcup_{i={ i_0}}^{\infty}B_i}$, we get that the overall measure of the regions with bad gradient bounds from some flow time $\lambda_0:=\frac{i_0}{K_{\Lambda}^2}$ onward is bounded by an arbitrarily small number as $\lambda_0\to\infty$: 
\begin{equation}\label{bad_est}
{|\mathcal{B}_{i_0}| \leq \frac{1}{K_\Lambda^2} {\frac{e^{1/9}}{e^{1/9}-1}}e^{-\frac{i_0}{9}}}\ . 
\end{equation}
Denoting by $G_i$ the complement of the  $B_i$, {\it {\it i.e.}}~the set of flow times with good gradient bounds: $G_i:={\frac{1}{K_\Lambda^2}}[i,i+1)-B_i$, we have that 
\begin{equation}\label{good_bound}
\int dz\; |\nabla K(\lambda,z)|^2 <   C_{12}\,K_\Lambda^{ 3} e^{-{\frac{2}{9}}i},\;\;\;\;\textrm{for every } \lambda\in G_i\ .
\end{equation}
Let us denote the total set of flow times with good gradient bound as ${\mathcal{G}:=\bigcup_{i={ 0}}^{\infty} G_i}$. 

We now claim the following Lemma about the spatial uniformity of $K$ at times when the gradient bounds are good:
\begin{lemma}\label{lemma:Kuniform}
 It exists a flow time $\lambda_{0,3}$ such that, for $\lambda> \lambda_{0,3}\geq 0$, in those flow times with good gradient bounds, $K$ is close to $K_\Lambda$, {\it {\it i.e.}}:
\begin{equation}\label{lower_K_bd}
{|K -K_{\Lambda}|\leq 2K_{\Lambda} { \sqrt{C_{12}}}\,e^{-\frac{1}{9}i}},\;\;\; \textrm{for every } \lambda \in  G_{ i}\cap\{\lambda|\lambda> \lambda_{0,3}\}\ . 
 \end{equation} 
 \end{lemma}
 \begin{proof}
By Cauchy-Schwartz, 
\begin{align}
&|K(z_2)-K(z_1)| \leq \left| \int_{z_1}^{z_2} |\partial_zK|dz\right| \leq (z_2-z_1)^{1/2}\left(\int dz\, |\nabla K|^2\right)^{1/2}\\\nonumber
&\qquad < (z_2-z_1)^{1/2} { \left(C_{12}\,K_\Lambda^{ 3}\right)^{1/2}} e^{-{\frac{1}{9}}i}\ , 
\end{align}
where in the last inequality, we have used that $\lambda \in G_i$. 
Thus, if there is a point, $z_1$, where \eqref{lower_K_bd} is violated, then 
\begin{equation}
|K - K_{\Lambda}| \geq K_{\Lambda}{ \sqrt{C_{12}}}\,e^{-\frac{1}{9}i}\ ,
\end{equation}
on the interval $[z_1,z_1+1{ /K_\Lambda}]$. {By using that $K+K_\Lambda\geq K_\Lambda$, this gives 
\begin{equation}
\int_{0}^{L(\lambda)}  dz\;  |K_{\Lambda}^2-K^{2}(\lambda,z)| \geq \sqrt{C_{12}}K_{\Lambda}e^{-\frac{1}{9}i}\ ,
\end{equation}
which, together with \eqref{z_quant_est} and the fact that ${K_\Lambda^2}\lambda\in [i,i+1)$ imply
\begin{equation}
\sqrt{C_{12}}K_{\Lambda}e^{-\frac{1}{9}i} \leq  C_{13}K_{\Lambda}e^{-\frac{1}{3}i}\ ,
\end{equation}
yielding the inequality $i\leq \frac{9}{2}\log\left(\frac{C_{13}}{\sqrt{C_{12}}}\right)$. Thus, taking

\be\label{eq:lambdazerothreemin}
\lambda_{0,3}=\frac{9}{2K_\Lambda^2}\log\left(\frac{C_{13}}{\sqrt{C_{12}}}\right)+\frac{1}{K_{\Lambda}^2}\ ,
\ee
there cannot be such $z_1$ for $\lambda>\lambda_{0,3}$, establishing  \eqref{lower_K_bd}.}
\end{proof}

As mentioned, our strategy now is to {study the spacetime metric using the MCF foliation}. Given a point $p\in \mathcal{M}_{\lambda}$, the metric of the four-dimensional spacetime at $p$ is given by
\be\label{eq:metric4fromMCF}
ds_4^2=g^{(4)}_{\mu\nu} dx^\mu dx^\nu=-K^2d\lambda^2+g_{ij} dx^i dx^j\ ,
\ee
where we remind that $g_{ij}$ is the 3-metric of the MCF {hypersurfaces}. This parametrization of $g^{(4)}$ is useful as long as the MCF foliates a large region of spacetime. This is indeed the case, as we are going to show next.

In the subset of $M^{(3+1)}$ that is foliated by the flow we have chosen a time function such that the lapse is $N=1$ (see the discussion below Definition \ref{def:crushing} and \cite{Creminelli:2019pdh}). Let $t_{\min}(\lambda)$ be the smallest value of $t$ in $\mathcal{M}_{\lambda}$. By Lemma \ref{Ham_trick}, $t_{\min}(\lambda)$ is  a locally Lipschitz function, hence differentiable almost everywhere, and at such point of differentiability 
\begin{equation}
\frac{d}{d\lambda}t_{\min}(\lambda)=\frac{\partial t}{\partial \lambda}(x_{\lambda},\lambda)\ ,
\end{equation}   
where $x_{\lambda}\in \mathcal{M}_{\lambda}$ is a point where $t_{\min}(\lambda)$ is attained. Note that, by minimality, $\nabla t \perp T_{x_{\lambda}}\mathcal{M}_{\lambda}$, and so
\begin{equation}
\frac{\partial t}{\partial \lambda}(x_{\lambda},\lambda)=g^{(4)}\left(\nabla t, \frac{dx_{\lambda}}{d\lambda}\right) =g^{(4)}\left(\nabla t, -K(x_{\lambda}) \nabla t \right)=K(x_{\lambda})\ .
\end{equation}
Note that, from (\ref{bad_est}) and (\ref{lower_K_bd}), we can choose $\lambda_{0,4}=\max\left(\lambda_{0,3},\frac{9}{K_\Lambda^2}\log\left(\frac{e^{1/9}}{e^{1/9}-1}\right), \frac{9}{K_\Lambda^2} \log(4\sqrt{C_{12}})\right)+\frac{1}{K_{\Lambda}^2}$, such that, for $\lambda\geq\lambda_{0,4}$, ${|\mathcal{B}_{i_{0,4}}|}\leq 1/K_\Lambda^2$  and, when $\lambda\in {\mathcal{G}}\cap [\lambda_{0,4},+\infty)$, $K\geq K_\Lambda/2$. {Here $i_{0,4}$ is the integer in the interval $[K_{\Lambda}^2\lambda_{0,4}-1,K_{\Lambda}^2\lambda_{0,4})$}. For $\lambda\geq\lambda_{0,4}$, we can therefore write
\bea\label{eq:tmin}
&&t_{\min}(\lambda)=t_{\min}(\lambda_{0,4})+\int_{[\lambda_{0,4},\lambda]} d\lambda' \;K(x_{\lambda'},\lambda')\geq\\ \nonumber
&&\quad\qquad\geq t_{\min}(\lambda_{0,4})+\int_{[\lambda_{0,4},\lambda] \cap {\mathcal{G}}} d\lambda'\;K(x_{\lambda'},\lambda')\geq  t_{\min}(\lambda_{0,4})+\frac{K_\Lambda}{2}\left(\lambda-\lambda_{0,4}-\frac{1}{K_\Lambda^2}\right)\ .
\eea
We therefore conclude that the flow reaches arbitrary large $t$ as $\lambda\to +\infty$, and therefore, since the time function has lapse equal to 1, it foliates arbitrarily large regions of the spacetime. This guarantees that the spacetime metric we constructed from (\ref{eq:metric4fromMCF}) is valid in such regions. 

{Additionally, this implies that $M^{(3+1)}$ has no crushing singularities. Indeed, if there were such a singularity, there exists a $c > c_0$ such that the flow never reaches $S_c$ as in Definition~\ref{def:crushing}.  Choose $c_1>c > c_0\geq 0$ in our time function as defined below  Definition \ref{def:crushing}. Let $p\in S_c$, certainly $t(p)<\infty$. Connecting $p$ to $M_0$ by a timelike curve, $t$ must grow monotonically on this curve and bounded above by $t(p)$. But ${\cal M}_\lambda$ intersects this curve for arbitrarily large $\lambda$'s since the flow does not reach $p$, contradicting that the minimum time on the flow slices grows arbitrarily large, $(\ref{eq:tmin})$. Since we are assuming that $M^{(3+1)}$ has only potential singularities of the crushing kind, this implies that $M^{(3+1)}$ has no singularities, and is therefore future-directed time-like and null geodesically complete.}

Now, let {$\gamma:[\lambda_0,\lambda_1]\rightarrow M^{(3+1)}$ be a smooth curve in $M^{(3+1)}$, with $\gamma(\lambda)\in \mathcal{M}_{\lambda}$, where $\lambda_0\geq \max\left(\lambda_{0,4},\lambda_\ast\right)$. Here $\lambda_{\ast}$ is from Theorem \ref{main_spatial_thm} and $\lambda_{0,4}$ is from the paragraph above. We are interested in comparing the metric in \eqref{eq:metric4fromMCF} with the model metric
\be\label{eq:metric4fromMCFdS}
ds^2_{\mathbf g^{(4)}} = {\mathbf g}^{(4)}_{\mu\nu} dx^\mu dx^\nu=-K_{\Lambda}^2d\lambda^2+\mathbf{g}_{ij} dx^i dx^j\ ,
\ee
where $\mathbf{g}$ is defined by \eqref{comp_metric}.}
We can estimate the difference in length of the curve $\gamma$ as measured with the {actual metric \eqref{eq:metric4fromMCF} and with the reference metric \eqref{eq:metric4fromMCFdS}}. Because of the Lorentzian nature of the spacetime, we will separately bound the difference of the evaluation of the contraction of the tangent vector with the $\lambda$-direction, and with the spatial direction. For the time direction, letting $i_0$ be the integer in the interval $[K_{\Lambda}^2\lambda_0-1,K_{\Lambda}^2\lambda_0)$, we can write 
\begin{align}\label{time_part_diff}
&\int_{{\lambda_0}}^{{\lambda_1}} d\lambda\; {\sqrt{|g^{(4)}_{\lambda\lambda}-{\mathbf g}^{(4)}_{\lambda\lambda}|}} {=} \int_{\lambda_0}^{\lambda_1}d\lambda\; {\sqrt{|K_{\Lambda}^2-K^2(\gamma(\lambda))|}}\\
&=\int_{{\mathcal{B}_{i_0}}\cap [\lambda_0,\lambda_1]} d\lambda\; {\sqrt{|K_{\Lambda}^2-K^2({ \gamma(\lambda)})|}}+\int_{{\mathcal{G}}\cap [\lambda_0,\lambda_1]} d\lambda\; {\sqrt{|K_{\Lambda}^2-K^2({ \gamma(\lambda)})|}} \\
&\leq\label{ineq_1} { K_{\Lambda} } |{ \mathcal{B}_{i_0}}|+\sum_{i={ i_0}}^{\infty}\int_{G_i\cap [\lambda_0,\lambda_1]}d\lambda\; {\sqrt{|K_{\Lambda}^2-K^2({ \gamma(\lambda)})|}} \\
& \leq\label{ineq_2}  {\frac{e^{1/9}}{e^{1/9}-1}{\frac{1}{K_\Lambda}}e^{-\frac{K_{\Lambda}^2}{9}{\lambda_0}}+4\sqrt[4]{C_{12}}{ \frac{1}{K_\Lambda}}\sum_{i={ i_0}}^{\infty}e^{-\frac{{ i}}{18}}  }\ .
\end{align}
Here, at \eqref{ineq_1}, we have used the bound $|K_{\Lambda}^2-K^2| \leq K_{\Lambda}^2$ to bound the first integrand, and, in~\eqref{ineq_2}, we used the estimate~\eqref{bad_est}  to bound $|{\mathcal{B}_{i_0}|}$, and \eqref{lower_K_bd} to bound the second integral.  
Thus
\begin{equation}\label{time_part_diff}
\int_{{\lambda_0}}^{{\lambda_1}} d\lambda\; {\sqrt{|g^{(4)}_{\lambda\lambda}-{\mathbf g}^{(4)}_{\lambda\lambda}|}} \leq \frac{C_{13}}{K_{\Lambda}} e^{-\frac{K_{\Lambda}^2}{18}\lambda_0}
\end{equation}
where $C_{13}={ {\frac{e^{1/18}}{e^{1/18}-1}}}\left(1+4\sqrt[4]{C_{12}}\right)$.

{Considering the integral of the projection of the tangent vector on ${\cal M}_\lambda$, Theorem \ref{main_spatial_thm} directly implies
\begin{align}\label{eq:curveonsurface}
&\int_{\lambda_0}^{\lambda_1}d\lambda\; {\sqrt{|{g}(\dot{\gamma},\dot{\gamma})-\mathbf{g}(\dot{\gamma},\dot{\gamma})|}} \leq  {\sqrt{C_{11}}}e^{-\frac{1}{{ 12}}K_{\Lambda}^2\lambda_0}\int_{\lambda_0}^{\lambda_1}d\lambda\;{\sqrt{\mathbf{g}(\dot{\gamma},\dot{\gamma})}}\ .
\end{align}

We can now obtain an expression of the length discrepancy of such a curve, when computed w.r.t to the true metric and the comparison one. Namely, combining \eqref{time_part_diff} and \eqref{eq:curveonsurface}, we get

\bea\label{len_comp}
&&\left|L^{ds_4^2}[\gamma]-L^{ds^2_{4,{\mathbf g}}}[\gamma]\right| \leq \int_{\lambda_0}^{\lambda_1}d\lambda\; \left|\sqrt{|K^2-g(\dot \gamma,\dot \gamma)|}-\sqrt{|K_\Lambda^2-\mathbf{g}(\dot \gamma,\dot \gamma)|}\right|=\\ \nonumber
&&=\int_{\lambda_0}^{\lambda_1}d\lambda\;\left|\sqrt{|(K^2-K_\Lambda^2)-(g(\dot \gamma,\dot \gamma)-\mathbf{g}(\dot \gamma,\dot \gamma))+(K_\Lambda^2-\mathbf{g}(\dot \gamma,\dot \gamma))|}-\sqrt{|K_\Lambda^2-\mathbf{g}(\dot \gamma,\dot \gamma)|}\right|\leq\\ \nonumber
&&\leq \int_{\lambda_0}^{\lambda_1}d\lambda\;\sqrt{\left|(K^2-K_\Lambda^2)-(g(\dot \gamma,\dot \gamma)-\mathbf{g}(\dot \gamma,\dot \gamma))\right|}\leq\\ \nonumber
&&\leq \int_{\lambda_0}^{\lambda_1}d\lambda\;\left(\sqrt{\left|(K^2-K_\Lambda^2)\right|}+\sqrt{\left|(g(\dot \gamma,\dot \gamma)-\mathbf{g}(\dot \gamma,\dot \gamma))\right|}\right)\leq \\ \nonumber
&&\leq \frac{C_{13}}{K_{\Lambda}} e^{-\frac{1}{18}K_{\Lambda}^2\lambda_0}+  {\sqrt{C_{11}}}e^{-\frac{1}{{ 12}}K_{\Lambda}^2\lambda_0}\int_{\lambda_0}^{\lambda_1}d\lambda\;{\sqrt{\mathbf{g}(\dot{\gamma},\dot{\gamma})}}\ ,
\eea
where for the second and third inequalities  we have used the triangle inequality and the inequality $\sqrt{a+b} \leq \sqrt{a}+\sqrt{b}$ for $a,b\geq 0$.

 \begin{figure}
\begin{center}
\includegraphics[width=10cm,draft=false]{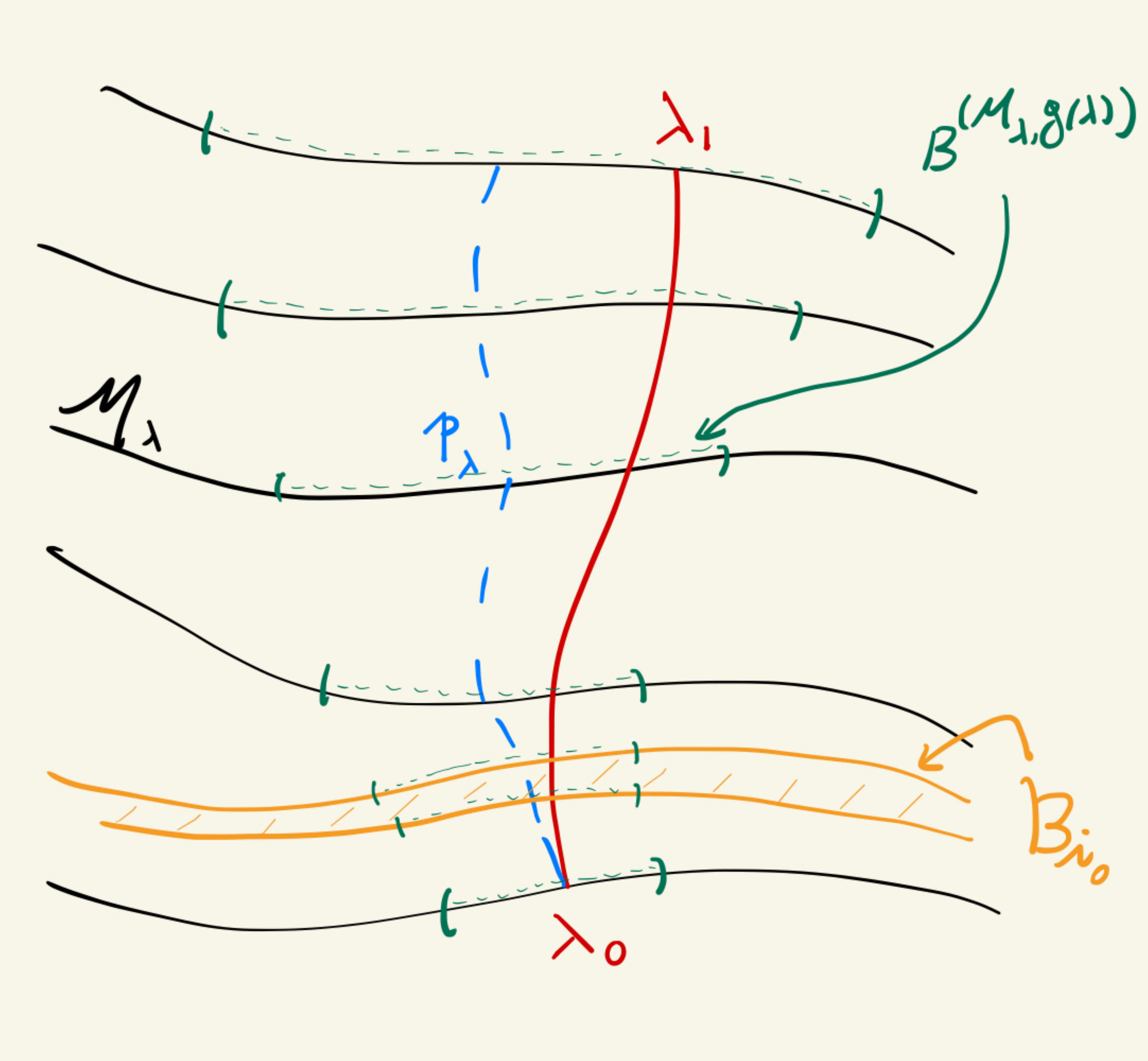}
\end{center}
\caption{\small
{Some of the geometric quantities that are defined in Sec.~\ref{space_time_sec}.} \label{fig:quantities_sec_7} }
\end{figure}

This means that for any curve $\gamma$ such that  $\gamma(\lambda)\in \mathcal{M}_{\lambda}$,  its length w.r.t.~the spacetime metric converges exponentially, as we take $\lambda_0$ larger and larger, to the respective quantity evaluated on the comparison metric ${\mathbf g}^{(4)}$.  { Note further that if such a curve is future-pointing timelike or null \textit{w.r.t the true spacetime metric $ds_4^2$}, then  $g(\dot{\gamma},\dot{\gamma}) \le  2K_{\Lambda}^2$, as below eq.~\eqref{eq:gradKproof}, so, by Theorem \ref{main_spatial_thm} {(provided $\lambda_0$ is large enough, say, as before, $\lambda_0 > \frac{6\log(100C_{11})}{K_{\Lambda}^2})$} 
\be\label{time_like_spat_speed}
\sqrt{\mathbf{g}(\dot{\gamma},\dot{\gamma})} {<} \sqrt{2} \sqrt{g(\dot{\gamma},\dot{\gamma})} \le 2K_{\Lambda}\ .
\ee 
This and \eqref{len_comp} yield
\begin{equation}\label{len_comp_conc1}
\left|L^{ds_4^2}[\gamma]-L^{ds^2_{{\mathbf g}^{(4)}}}[\gamma]\right| \leq \frac{C_{13}}{K_{\Lambda}} e^{-\frac{1}{18}K_{\Lambda}^2\lambda_0}+ { 2}\sqrt{C_{11}}K_\Lambda e^{-\frac{1}{12}K_{\Lambda}^2\lambda_0} \left(\lambda_1-\lambda_0\right)\ .
\end{equation}

In fact,  any future-pointing timelike or null curve {\textit{w.r.t the true spacetime metric $ds_4^2$},  with non-vanishing velocity}, can be re-parametrized so that $\gamma(\lambda)\in \mathcal{M}_{\lambda}$. Such a re-parametrization is possible since {if $u_0$ is a critical point of the function $\lambda(\gamma(u))$ then  $\dot{\gamma}(u_0)\in \mathcal{M}_{\lambda}$}, so the tangent to $\gamma$ is spacelike at such a point.
As the length of curves is invariant under re-parametrization,  the length of all future-pointing timelike or null curves becomes very close to the length computed with the model metric \eqref{eq:metric4fromMCFdS}. \\}

{ 
To compare lengths with the exact de Sitter space, we need {to use Theorem \ref{main_spatial_thm_2}. To do so we need to prove that time-like and null curves remain inside the ball where the theorem applies. This is given by} the following simple lemma {(see Fig.~\ref{fig:quantities_sec_7})}.

\begin{lemma}\label{lemma:gammaball}
Let $\gamma:[\lambda_0,\lambda_1]\rightarrow M^{(3+1)}$ be a smooth curve in $M^{(3+1)}$ where
\be\label{ineqballs}
\lambda_0\geq \max\left(\lambda_{\ast},\frac{6\log(100C_{11})}{K_{\Lambda}^2}\right),
\ee
 and such that $\gamma(\lambda)\in \mathcal{M}_{\lambda}$.  Assume further that $\gamma$ is timelike or null, and denote $p=\gamma(\lambda_0)$ and $p_{\lambda}$ the evolution of $p$ along the flow. Then, for each $\lambda$, 
\begin{equation}
\gamma(\lambda)\in B^{(\mathcal{M}_{\lambda},g(\lambda))}\left(p_{\lambda}, \frac{{ 12}}{K_{\Lambda}}e^{\frac{1}{3}K_{\Lambda}^2(\lambda-\lambda_0)}\right).
\end{equation}
\end{lemma}
\begin{proof}
Consider $\mu \in \mathcal{M}_{\lambda_0}$ be the curve obtained by, for each $\lambda$, following $\gamma(\lambda)$ by MCF back to $\mathcal{M}_{\lambda_0}$.  Then 
\begin{equation}\label{eq:firstcurve}
L^{\mathbf{g}}[\mu]=\int_{\lambda_0}^{\lambda_1}d\lambda\;\sqrt{\mathbf{g}(\lambda_0)(\dot{\mu},\dot{\mu})}=\int_{\lambda_0}^{\lambda_1}d\lambda \;e^{-\frac{1}{3}K_{\Lambda}^2(\lambda-\lambda_0)}\,\sqrt{\mathbf{g}(\lambda)(\dot{\gamma},\dot{\gamma})}\leq \frac{{6}}{K_{\Lambda}}\ , 
\end{equation} 
where, in the last inequality, we have used \eqref{time_like_spat_speed} {(which requires both inequalities \eqref{ineqballs})}. Therefore, for every $\lambda'\in [\lambda_0,\lambda_1]$, letting $\mu^{\lambda'}$ be the curve obtained by following $\gamma(\lambda)$ by MCF forward to $\mathcal{M}_{\lambda'}$  for each $\lambda\in [\lambda_0,\lambda']$, we get 
\begin{equation}\label{eq:secondcurve}
L^{\mathbf{g}}[\mu^{\lambda'}]=\int_{\lambda_0}^{\lambda'}d\lambda\sqrt{\mathbf{g}(\lambda')(\dot{\mu}^{\lambda'},\dot{\mu}^{\lambda'})}\leq e^{\frac{1}{3}K_{\Lambda}^2(\lambda'-\lambda_0)}\int_{\lambda_0}^{\lambda_1}d\lambda\sqrt{\mathbf{g}(\lambda_0)(\dot{\mu},\dot{\mu})}\leq  \frac{{ 6}}{K_{\Lambda}}e^{\frac{1}{3}K_{\Lambda}^2(\lambda'-\lambda_0)} \ ,
\end{equation} 
{where in the last passage we used~(\ref{eq:firstcurve}).}
Since on a given flow time slice, $\mathbf{g}$ and $g$ lengths are close to each other (by Theorem \ref{main_spatial_thm}{, which requires both inequalities \eqref{ineqballs}}), we get that   
\be
L^{g}[\mu^{\lambda'}] \leq  \frac{{ 12 }}{K_{\Lambda}}e^{\frac{1}{3}K_{\Lambda}^2(\lambda'-\lambda_0)}.
\ee
Since $\mu^{\lambda'}$ is a curve in $\mathcal{M}_{\lambda'}$ connecting $p_{\lambda'}$ with $\gamma(\lambda')$ of length $\le \frac{{ 12}}{K_{\Lambda}}e^{\frac{1}{3}K_{\Lambda}^2(\lambda'-\lambda_0)}$, the result follows.
\end{proof}
Now, let $\gamma$ and $\lambda_0$ be as in the above lemma, and assume further that
\be
\lambda_0 \geq \max\left(\lambda_{\ast\ast},\lambda_{0,4},\frac{12\log({ 12})}{K_{\Lambda}^2}\right)\ ,
\ee
where $\lambda_{\ast\ast}$ is from Theorem \ref{main_spatial_thm_2}, {$\lambda_0 > \lambda_{0,4}$ guarantees the validity of the metric \eqref{eq:metric4fromMCF} over large spacetime regions, see \eqref{eq:tmin}, and} $\lambda_0\geq \frac{12\log({ 12})}{K_{\Lambda}^2}$ ensures that the balls of Lemma~\ref{lemma:gammaball} are contained in the balls of applicability of Theorem \ref{main_spatial_thm_2}.  Therefore, the  Lemma above and Theorem \ref{main_spatial_thm_2}
 imply that 
 \begin{equation}
 ||g{(\lambda)}-\mathbf{g}_{\mathrm{dS}}{(\lambda)}||_{ g(\lambda)}\leq {16} e^{-\frac{1}{12}K_{\Lambda}^2\lambda_{0}}
 \end{equation}
 along $\gamma$, where $\mathbf{g}_{\mathrm{dS}}$ is given by \eqref{slice_ds_def}, defined using the point $p=\gamma(\lambda_0)$. Setting the space-time \textit{exact} de Sitter metric,
\begin{equation}
ds^2_{\rm dS}:= \mathbf{g}_{\mathrm{dS}}^{(4)}:= -K_{\Lambda}^2d\lambda^2+(\mathbf{g}_{\mathrm{dS}})_{ij}dx^idx^j.
\end{equation}
Arguing as in \eqref{len_comp}, \eqref{time_like_spat_speed} and \eqref{len_comp_conc1}, we get that for every future-pointing timelike or null curve $\gamma:[a,b]\rightarrow M^{(3+1)}$ with $\lambda_0:=\lambda(\gamma(a))\geq \max\left(\lambda_{0,4},\lambda_{\ast\ast},\frac{12\log({ 12})}{K_{\Lambda}^2}\right)$, setting $\lambda_1=\lambda(\gamma(b))$ we get 
\begin{equation}
\left|L^{ds_4^2}[\gamma]-L^{ds^2_{\rm dS}}[\gamma]\right| \leq \frac{C_{13}}{K_{\Lambda}} e^{-\frac{1}{18}K_{\Lambda}^2\lambda_0}+ { 8}K_\Lambda e^{-\frac{1}{24}K_{\Lambda}^2\lambda_0} \left(\lambda_1-\lambda_0\right).
\end{equation}}}

We therefore conclude that the length of any future-oriented, timelike or null curve between two points converges exponentially fast to the same quantity evaluated with the de Sitter metric, as we take the lowest time of the two points, $\lambda_0$, larger and larger.

\section{Dilution of Matter}\label{dilution_sec}

We now show that the stress tensor goes to zero almost everywhere. We can bound  the integral over $z$ of $|T_{\mu\nu}n^\mu n^\nu|$. One can use eq.~(\ref{eq0}) and the WEC to write
\bea
&&16\pi G_N\int dz\; |T_{\mu\nu}n^\mu n^\nu|=16\pi G_N \int dz\; T_{\mu\nu}n^\mu n^\nu=\\ \nonumber
&&\qquad=\int dz\; \left({^{(3)}\!R} + \frac{2}{3} \left(K^2-K_\Lambda^2\right) - \sigma^2 \right)\leq C_{14} K_\Lambda e^{-\frac{1}{3}  K_\Lambda^2\lambda}\ ,
\eea
where in the last step we used the bounds (\ref{first_L}) and (\ref{quant_eq}) together with Theorem~\ref{length_growth_lemma}.
We defined $C_{14}:=(C_5+C_6) (1+\delta)K_\Lambda L(0)$.

Because of the DEC, $T_{\mu\nu} n^\mu n^\nu$ is at least as large as the absolute value of any other component of the stress tensor in an orthonormal frame where $n^\mu$ is the timelike vector~\footnote{This is actually an equivalent definition of the DEC \cite{Hawking:1973uf} as it is straightforward to verify.}. We therefore define a vierbein $e_\mu{}^{a}$, such that $g_{\mu\nu}^{(4)}=e_\mu{}^{a}e_\nu{}^{b}\eta_{ab}$, with $\eta_{ab}$ being the Minkowski metric. We choose $e_\mu{}^0=n_\mu$. By DEC, we have
 \be\label{eq:Tmunubound}
  16\pi G_N\int dz\;  \left| T_{\mu\nu} e^{\mu a}e^{\nu b}\right|\leq 16\pi G_N \int dz\; T_{\mu\nu}n^\mu n^\nu \leq  C_{14} K_\Lambda e^{-\frac{1}{3}  K_\Lambda^2\lambda}\ .
 \ee
Since, by the symmetries of the problem, $T_{\mu\nu}$ is uniform on the slices at constant $z$, we see that in almost-all of the ever-growing $z$-direction, $ G_NT_{\mu\nu}$ has to be at most of order~$K_\Lambda^2 \cdot {\cal {O}}(e^{-\frac{2}{3}K_\Lambda^2\lambda})\to 0$, while it can be of order $K_\Lambda^2$ only on a shell of $z$-thickness that shrinks as $e^{-\frac{1}{3}K_\Lambda^2\lambda}$ (or even faster if $T_{\mu\nu}$ gets larger) and therefore this shell is just a fraction of order $e^{-\frac{2}{3}K_\Lambda^2\lambda}$ of the extension of the $z$ direction. 

Notice that, by Einstein's equations, this means that a similar bound applies to $R_{\mu\nu}$.  In fact, we can take the Einstein equations and contract them with $e^{\mu a}e^{\nu b}$ 
  \bea
&&  {R}_{\mu\nu} e^{\mu a}e^{\nu b} =\left[8\pi G_N \left(T_{\mu\nu}- \frac{g_{\mu\nu}}{2} T\right)+\frac{1}{3}K_\Lambda^2 g_{\mu\nu}\right] e^{\mu a}e^{\nu b}\ .
 \eea
 Let us write $R_{\mu\nu}$ as $R_{\mu\nu}=R_{dS,\mu\nu}+\delta R_{\mu\nu}$, where $R_{dS,\mu\nu}=\frac{1}{3} K_\Lambda^2 g_{\mu\nu}$  is the Ricci tensor of de Sitter space with cosmological constant $\Lambda$. We obtain
 \bea
&&  {\delta R}_{\mu\nu}e^{\mu a}e^{\nu b} = \;8\pi G_N \left(T_{\mu\nu}- \frac{g_{\mu\nu}}{2} T\right)e^{\mu a}e^{\nu b} \ .
 \eea
We can now use the bound (\ref{eq:Tmunubound}) to write
  \bea
&& \int dz \;\left|  {\delta R}_{\mu\nu}e^{\mu a}e^{\nu b}\right| =\int dz \;8\pi G_N \,\left| T_{\mu\nu}e^{\mu a}e^{\nu b}-T \eta^{ab}\right|\leq \frac{3}{2} C_{14} K_\Lambda e^{-\frac{1}{3}  K_\Lambda^2\lambda}  \ .
 \eea

It is hard to imagine that  one can achieve a control on $T_{\mu\nu}$ which is better than this, without additional assumptions on the stress tensor and using arguments similar to the ones  presented in~\cite{Creminelli:2019pdh}. In particular one cannot hope for a pointwise convergence of the stress tensor (and thus of the Ricci tensor), since it is easy to come up with counterexamples. Indeed, one can imagine an alien population  living in spaceships and whose main purpose in life is to prevent pointwise convergence to de Sitter space.  While, by the symmetries of the problem, these aliens are constrained to be uniformly distributed on expanding surfaces, nothing prevents them from squeezing their spaceships fast enough in the $z$-direction, in order to keep the energy density constant in their surface-like ships. Therefore the stress tensor and the Ricci tensor do not need to go to zero everywhere. Furthermore, no physical law seems to prevent these aliens from splitting each of their spaceships into smaller ones at each Hubble time, $1/K_\Lambda$, creating thinner spaceships but keeping constant their energy density. In doing so and distributing the spaceships in the $z$-direction one can always have one spaceship in each region in the $z$-direction of size $\sim 1/K_\Lambda$. Thus one in general does not have pointwise convergence in any large portion of space.

The fact that $T_{\mu\nu}$ does not converge pointwise is not in contradiction with the pointwise convergence of the spatial metric~\footnote{This is peculiar of the setup we are discussing, where $T_{\mu\nu}$ can only depend on $z$. In a generic case without symmetries, a point-like localised mass, no matter how small, will affect the metric if one goes sufficiently close to it.}. For instance, if one considers an infinitesimally thin layer of matter localised at a certain value of $z$, the solution of the Einstein equations across this thin wall gives the so-called Israel junction conditions~\cite{Israel:1966rt}. The metric of this 2+1 dimensional surface is continuous across the wall and the jump in the extrinsic curvature of the wall, $K_{\alpha\beta}^+ - K_{\alpha\beta}^-$, is fixed by the surface stress tensor $S_{\alpha\beta}$ (the stress tensor integrated over a small interval in $z$ across the wall):
\be\label{Israel}
K_{\alpha\beta}^+ - K_{\alpha\beta}^- = 8 \pi G_N \left(S_{\alpha\beta}- \frac{g_{\alpha\beta}}{2} g^{\gamma\delta}S_{\gamma\delta}\right) \;.
\ee
(The indices $\alpha, \beta, \ldots$ span the (2+1)-dimensional space at fixed $z$ and $g_{\alpha\beta}$ is the induced metric on this space.)
The expansion of the thin wall in the directions orthogonal to $z$ will make the surface stress tensor go to zero, so that also the jump in the extrinsic curvature vanishes asymptotically, in agreement with the pointwise bound \eqref{H_bound}, which applies to the components of the extrinsic curvature on ${\cal M}_\lambda$. In particular one can check that when the thin wall saturates the SEC,  so that its surface stress tensor goes to zero as slowly as possible within our assumptions, the bound \eqref{H_bound} is also saturated, as expected~\footnote{An isotropic surface stress tensor, $S_{ij} = diag(\sigma, \Pi, \Pi)$, saturates the SEC if $\Pi = -1/2 \cdot \sigma$. This can be understood starting from an object with a finite extension in the $z$ direction. One can prove, using the conservation of the stress energy tensor (see for instance \cite{Hui:2010dn}), that $\int dz \,T_{zz} =0$, independently of the internal dynamics of the wall.  In $3+1$ dimensions for a diagonal stress tensor the SEC implies $\rho +p_i \ge 0$ and $\rho + \sum_i p_i \ge 0$, where $p_i$ are the pressures in the three spatial directions. If we now apply this to the integral over $z$ of the stress tensor we obtain the limit the saturates the SEC.
In de Sitter space the surface energy density dilutes as a consequence of the conservation of the stress tensor: $\dot\sigma = - 2 \cdot \frac{K_\Lambda}{3} (\sigma+ \Pi) = -\frac13 K_\Lambda \sigma$, when SEC is saturated.  This gives $\sigma \propto \exp(-\frac13 K_\Lambda^2 \lambda)$. Using \eqref{Israel}, this is indeed the same behaviour as the pointwise bound \eqref{H_bound}.}.

\section{Summary and Physical Equivalence to de Sitter}\label{equiv_de_sit_sec}

\paragraph{Summary:} We have considered 3+1 dimensional cosmologies satisfying the Einstein equations with a positive cosmological constant and matter satisfying the dominant and the strong energy conditions. We have assumed that the  only potential singularities are of the crushing kind, and that the spatial slices have homogeneous but {potentially} anisotropic 2-surfaces. 
{We used the mean curvature flow to probe the geometry: spacetime is foliated by the mean curvature flow surfaces and the flow parameter runs orthogonal to them. We proved that} the spatial part of the resulting metric converges pointwise to the one of de Sitter space in flat slicing {on balls whose radius becomes arbitrarily large, growing as $e^{\frac{1}{3}K_\Lambda\lambda}$, as the flow time $\lambda$ goes arbitrarily large.} The lapse function converges to the one of de Sitter almost everywhere. The gradient of the lapse function converges to zero almost everywhere only once averaged over an arbitrarily small, but non-vanishing, time. We have then shown that these results imply that the length of any future-oriented, timelike or null curve  between two points at late enough time converges exponentially to the same quantity computed with the de Sitter metric. We have also shown that all components of the stress tensor go to zero almost everywhere.   Let us now explain in which sense our findings imply physical equivalence to de Sitter space at late enough times.

\paragraph{Physical Equivalence to de Sitter Space:}  Let us start by discussing the role of the residual matter, which, by~(\ref{eq:Tmunubound}), does not necessarily go to zero pointwise. However, the fact that {future-oriented} null geodesics, at late enough times, behave as in de Sitter space tells us that at late times there is a cosmological horizon approaching the one of de Sitter space. Therefore, fixing a late enough time $\lambda_2$, an observer will  be able to gather  information in the future only from points that, at $\lambda_2$, are contained in a ball, $B_c(\lambda_2) \subset {\cal M}_{\lambda_2}$, of radius $4 \cdot 3/K_\Lambda$; the de Sitter horizon is $3/K_\Lambda$. (The extra factor of 4 is included to account for the difference between the actual size of the horizon and the one of de Sitter space and also for the motion of the observer. These corrections decay exponentially in $\lambda_2$, and we are taking $\lambda_2$ late enough.) At any time $\lambda\geq\lambda_2$, the integral on ${\cal M}_\lambda\cap y_\lambda(y^{-1}_{\lambda_2}(B_c({\lambda_2})))$ of any component of the stress tensor in an orthonormal frame, is bounded by
\bea\nonumber
&& 16\pi G_N\int_{{\cal M}_\lambda\cap y_\lambda(y^{-1}_{\lambda_2}(B_c({\lambda_2})))}\!\!\! |T_{\mu\nu}e^{\mu a} n^{\nu b}|\leq 16\pi G_N\int_{{\cal M}_\lambda\cap y_\lambda(y^{-1}_{\lambda_2}(B_c({\lambda_2})))}\!\!\! T_{\mu\nu}n^\mu n^\nu\leq\\  
 &&\qquad \leq\frac{\pi (12)^2  C_{14}}{K_\Lambda}  e^{-\frac{1}{3}  K_\Lambda^2\lambda}\leq\frac{\pi(12)^2  C_{14}}{K_\Lambda}  e^{-\frac{1}{3}  K_\Lambda^2\lambda_2} \ ,
\eea
where we used (\ref{eq:Tmunubound}) at time $\lambda$. We therefore see that the overall energy and momentum contained at any time $\lambda\geq\lambda_2$ in the ball of points that are causally connected to the center goes to zero as we send $\lambda_2\to +\infty$. Since any experiment has some finite energy or momentum threshold below which no measurement can be done, we conclude that the residual matter content is equivalent to vacuum for all physical purposes.

Let us now discuss in what sense our results show that the geometry is physically the same as the one of de Sitter space. We have shown that {future-oriented timelike and null} geodesics converge to the ones of de Sitter. The equivalence principle states that free-falling particles follow geodesics {of this kind}, so that from this point of view the spacetime is effectively asymptotically de Sitter. However the equivalence principle is only a low energy approximation: particles can be directly coupled to the Riemann tensor (consider for instance a coupling of a scalar field $\phi$ of the form $\int d^4x\sqrt{-g^{(4)}} \;R^{\mu\nu\rho\sigma}\partial_\mu\partial_\rho\phi\partial_\nu\partial_\sigma\phi/\Lambda_{\rm HE}^{4}$ with $\Lambda_{\rm HE}$ being some high-energy scale) and we do not have control of the Riemann tensor. This kind of effects are suppressed at low energy by powers of the ratio of the energy scale of the experiment over $\Lambda_{\rm HE}$: at long enough distances they can be neglected. Therefore the equivalence with de Sitter space holds in the low-energy regime, when the effects that violate the equivalence principle can be neglected. {On top of this, on extremely large distances, larger than a ball whose radius grows as $e^{\frac{1}{3}K_\Lambda^2\lambda}$, with $\lambda$ arbitrarily large, the geometry is indeed not the one of de Sitter, but, since there is a cosmological horizon, these are causally disconnected regions and a local observer cannot experience this departure from de Sitter~\footnote{If, instead of a cosmological constant, we had an inflationary field, the approximately de Sitter phase would end at some time, and sufficiently long time after that moment, these long distance regions would become observable again.}}.

\paragraph{Outlook:} We have offered a proof of a de Sitter no-hair theorem in 3+1 dimensions for the case where the spacetime manifold has spatial slices that can be foliated by {2-dimensional} surfaces that are the closed orbits of a symmetry group.  Concerning the inflationary  `initial patch problem', these results, together with the ones that we discussed in the introduction, and in particular the numerical ones, substantially resolve it: one does not need quasi homogeneous initial conditions on a volume whose linear size is of the order of the Hubble radius of the  inflationary solution for inflation to start. 

Clearly, it would be nice to get rid of some of the symmetry assumptions we made here,  to consider initial surfaces that are not expanding everywhere, and {to include in the setup a dynamical inflaton}. Work is in progress in these directions \cite{inprogress}. 

\section*{Acknowledgements}

We thank Matt Kleban, Shamit Kachru, Brett Kotschwar, Jonathan Luk, Richard Schoen for discussions. OH has been partially supported by a Koret Foundation early career scholar award. LS is partially supported by Simons Foundation Origins of the Universe program (Modern Inflationary Cosmology collaboration) and LS by NSF award 1720397. AV is partially supported by NSF award DMS-1664683. PC would like to thank the Stanford Institute for Theoretical Physics for hospitality and support during part of this work. LS and AV would like to thank the International Center for Theoretical Physics for hospitality and support during part of this work.

\begin{small}
 
\bibliography{references}


\end{small}

\end{document}